%% file: OSUMS.tex
\documentclass[acmlarge]{acmart}

\usepackage{booktabs} 
\usepackage{hyperref}   
\usepackage{url}  

\usepackage{multirow}  
\usepackage{color}  

\usepackage{diagbox}
\usepackage{amsmath}
\usepackage{arydshln} 
\usepackage{color} 

\usepackage[linesnumbered,ruled]{algorithm2e} 

\SetAlFnt{\small}
\SetAlCapFnt{\small}
\SetAlCapNameFnt{\small}
\SetAlCapHSkip{0pt}
\IncMargin{-\parindent}

\acmJournal{TKDD}
\acmVolume{1}
\acmNumber{1}
\acmArticle{1}
\acmYear{2021}
\acmMonth{1}
\acmArticleSeq{1}

\setcopyright{acmcopyright}

\acmDOI{0000001.0000001}

\received{XX 2020}
\received{XX 2020}
\received[XX]{2021}

\begin{document}
\title{On-shelf Utility Mining of Sequence Data}

\author{Chunkai Zhang}

\affiliation{
	\institution{Harbin Institute of Technology (Shenzhen)}
	\city{Shenzhen}
	\country{China}
}
\email{ckzhang@hit.edu.cn}

\author{Zilin Du}
\affiliation{
	\institution{Harbin Institute of Technology (Shenzhen)}
	\city{Shenzhen}
	\country{China}
}
\email{yorickzilindu@gmail.com}

\author{Yuting Yang}
\affiliation{
	\institution{Harbin Institute of Technology (Shenzhen)}
	\city{Shenzhen}
	\country{China}
}
\email{20S151137@stu.hit.edu.cn}

\author{Wensheng Gan}
\authornote{This is the corresponding author}
\affiliation{
	\institution{Jinan University}
	\city{Guangzhou}
	\country{China}
}
\email{wsgan001@gmail.com}

\author{Philip S. Yu}
\affiliation{%
	\institution{University of Illinois at Chicago}
	\city{Chicago}
	\country{USA}
}
\email{psyu@uic.edu}

\begin{abstract}
	
	Utility mining has emerged as an important and interesting topic owing to its wide application and considerable popularity. However, conventional utility mining methods have a bias toward items that have longer on-shelf time as they have a greater chance to generate a high utility. To eliminate the bias, the problem of on-shelf utility mining (OSUM) is introduced. In this paper, we focus on the task of OSUM of sequence data, where the sequential database is divided into several partitions according to time periods and items are associated with utilities and several on-shelf time periods. To address the problem, we propose two methods, OSUM of sequence data (OSUMS) and OSUMS$^{+}$, to extract on-shelf high-utility sequential patterns. For further efficiency, we also designed several strategies to reduce the search space and avoid redundant calculation with two upper bounds time prefix extension utility (\textit{TPEU}) and time reduced sequence utility (\textit{TRSU}). In addition, two novel data structures were developed for facilitating the calculation of upper bounds and utilities. Substantial experimental results on certain real and synthetic datasets show that the two methods outperform the state-of-the-art algorithm. In conclusion, OSUMS may consume a large amount of memory and is unsuitable for cases with limited memory, while OSUMS$^{+}$ has wider real-life applications owing to its high efficiency.
	
\end{abstract}

%
%




%
%


\maketitle

\renewcommand{\shortauthors}{C. Zhang et al.}

\input{1_introduction.tex}
\input{2_relatedwork.tex}
\input{3_preliminaries.tex}
\input{4_technic.tex}

\input{5_experiment.tex}

\input{6_conclusion.tex}

\section{Acknowledgment}
This research is funded by National Natural Science Foundation of China (Grant No. 62002136), Natural Science Foundation of Guangdong Province, China (Grant No. 2020A1515010970), and Shenzhen Research Council (Grant No. GJHZ20180928155209705).

\bibliographystyle{ACM-Reference-Format}
\bibliography{OSUMS.bib}
\end{document}

%% file: 1_introduction.tex
\section{Introduction}

It is well known that the tasks of association rule mining \cite{zhang2019privacy}, clustering \cite{huang2019ultra}, classification \cite{pavlinek2017text}, and prediction \cite{wu2018hybrid}, which primarily aim at extracting interesting information and patterns from data repositories, perform a crucial role in the domain of knowledge discovery in database (i.e., data mining) \cite{han2011data}. Among them, association rule mining concentrates predominantly on the existence of a causal relationship in databases within the framework of support and confidence, where a subproblem of frequent itemset mining (FIM) exists \cite{agrawal1994fast}. As the name implies, the goal of FIM is to discover all frequent itemsets with respect to a threshold called support, which is the lower limit of occurrences. As a significant part in the domain of pattern mining, the research issue of FIM has been extensively studied so that the patterns, i.e., frequent itemsets, can be extracted easily and efficiently. However, FIM cannot tackle sequence data, where each item is embedded in a timestamp. The inherent sequential order among items results in a serious combinatorial explosion of the search space that results in not only a significantly long execution time but also large memory consumption. To handle this issue, a series of sequential pattern mining (SPM) algorithms were proposed \cite{agrawal1995mining, ayres2002sequential}. Given a sequence database, the mining objective of the SPM algorithms is to extract the complete set of frequent subsequences as sequential patterns. Bearing a striking resemblance to FIM, SPM identifies patterns with the measure of frequency. Therefore, both FIM and SPM approaches belong to the frequency-oriented framework. More details about FIM and SPM can be obtained from \cite{fournier2017surveyi, fournier2017surveys}.

Evidently, the implicit assumption within the frequency-oriented framework is that the frequency of a pattern, as the measure, represents its significance and interest; however, the frequency of a pattern cannot completely reveal the interesting aspects of a pattern \cite{gan2020fast}. There are several measures, such as utility, diversity, conciseness, applicability, and surprisingness that should be considered while extracting and ranking patterns with respect to their potential interests to the user \cite{geng2006interestingness}. Among them, utility is a semantic measure, that is, it considers the objective factors of the data itself, as well as the preferences of the user. In general, this type of interestingness is based on user-specified utility functions that can fit the importance of data, which are generated from domain knowledge or customized expectation of users. In practice, in a retail store, decision makers prefer goods with a high return on investment to those merely sold well, which is based on achieving profit maximization. For instance, the sales volume of televisions may be relatively low when compared to that of pencils; however, even certain infrequent sales patterns that include televisions may yield a higher profit than those pertaining to pencils. Hence, in this circumstance, the measure of patterns is supposed to consider a tradeoff between sales volume and unit profit to determine the preference of users. Thus, the concept of utility was introduced into data mining, resulting in the inception of the field of utility mining \cite{gan2019utility,gan2020proum}. As \cite{gan2018surveyuo} formulated, all utility mining methods belong to the utility-oriented framework, which considers not only the quantity (i.e., internal utility) but also relative importance (i.e., external utility) of items. Accordingly, FIM and SPM have been generalized as the tasks of high-utility itemset mining (HUIM) \cite{gan2018surveyii} and high-utility SPM (HUSPM) \cite{truong2019survey}, respectively, by incorporating the concept of utility. When compared to HUIM, HUSPM considers the sequential ordering of items and has to encounter a more serious combinatorial explosion of the search space, which results in a greater challenge in discovering desired patterns. Till now, several previous studies have developed a number of strategies and efficient data structures so that the HUSPM methods can be successfully applied to numerous real-life situations such as mobile commerce environment \cite{shie2011mining}, web browsing analysis \cite{ahmed2010mining}, and gene regulation \cite{zihayat2017mining}. 

It is obvious that HUSPM algorithms can identify interesting and informative patterns from complex databases because the problem of HUSPM simultaneously considers frequency, sequential order, and utility. However, HUSPM only considers relative time, i.e., sequential order of items in a sequence, and neglects absolute time in the database \cite{lan2014discovery}. In reality, all items in the database only appear in multiple different and user-specified time periods, namely on-shelf time. In real-world situations, certain goods are only placed on the shelf during certain short time periods; for example, in most retail stores, T-shirts are sold only in the summer and are removed from the shelf in other seasons. Therefore, the pattern $<$\textit{pencil}, \textit{pencil sharpener}$>$ is on the shelf throughout the year and more likely to be of high utility, while the pattern $<$\textit{T-shirt}, \textit{shorts}$>$ may demonstrate a low utility even though the items are sold well during their short on-shelf time. Obviously, conventional utility mining methods have a bias toward items that have longer on-shelf time as they have a higher chance of generating a high utility. To eliminate the bias, the problem of on-shelf utility mining (OSUM) was proposed, where the database is divided into several partitions according to time periods and each item is associated with several on-shelf time periods. Incorporating on-shelf information into the original measure, that is utility, is crucially beneficial to discover interesting and informative time-related patterns. On the one hand, consumers expect to buy their favorite goods at any time whether they are on sale or not, and they may blacklist the store even if the goods run out of stock on the shelves only a few times. On the other hand, retailers value every row on their shelves, where each inch of space is significantly valuable; moreover, they have to ensure that the right brands and goods are placed on the shelf at the appropriate times. Therefore, these extracted time-related patterns are helpful for retailers to make decisions. 

To fulfill the task of OSUM, HUIM and HUSPM are extended to the problems of OSUM in transaction data \cite{lan2014shelf, radkar2015mining, fournier2015foshu} and sequence data \cite{lan2014discovery}, respectively, while considering the on-shelf time. In the task of OSUM in transaction data, algorithms can identify on-shelf high-utility itemsets (osHUIs) by recursively enumerating itemsets in an ascending alphabetical order. For further efficiency improvement, a series of upper bounds as well as pruning strategies were developed to reduce the massive search space. To address the sequence data, Lan et al. \cite{lan2014discovery} proposed a novel approach serving as the first solution to discover on-shelf high-utility sequential patterns (osHUSPs). This issue of OSUM of sequence data considers a vital character, that is timestamps embedded in the items; thus, a method to determine osHUSPs has to overcome certain technical challenges, which are as presented below.

First, the calculation of utility is more complex than that of frequency as a subsequence may occur in a sequence multiple times. Therefore, knowing whether a subsequence will appear in a sequence is insufficient; determining all occurrences and choosing a proper calculation method are required.

Second, according to the Apriori property, the measurement frequency in the frequency-oriented framework is anti-monotonic. However, this download closure property is not satisfied for patterns with embedded utility values, which implies that it is impracticable to reduce the search space relying on the anti-monotone properties of the frequency-oriented framework in utility mining.

Third, a common and efficient method to mine patterns is by enumerating patterns in an alphabetical order recursively, which presents a general issue, that is, critical combinatorial explosion of the search space. This is because the intrinsic sequential ordering leads to several possibilities of concatenation, which is the operation of generating longer candidates. Thus, it is necessary to design tight upper bounds and powerful pruning strategies to overcome inefficient checking of a large number of candidates. 

Fourth, the search space in the problem sharply increases when the on-shelf information is incorporated. This is because the method must enumerate sequences in different time periods. Moreover, it is not straightforward to adapt the pruning techniques in HUSPM for OSUM of sequence data, and only a few pruning strategies were developed to address this problem.

To the best our knowledge, only one work \cite{lan2014discovery} focuses on extracting osHUSPs in sequence data. The development of the research topic is not yet mature, and the only existing algorithm, referred to as a two-phase approach for mining high on-shelf utility sequential patterns (TP-HOUS) has significant room for improvement in terms of execution time, memory consumption, candidates filtering, and scalability. Moreover, no systematic problem statement has been formulated. In this paper, we formulate the problem of OSUM of sequence data and propose two efficient algorithms. The major contributions of this study can be summarized as follows:

\begin{itemize}
	\item	We formulated the problem of OSUM of sequence data as discovering the complete set of osHUSPs. In particular, important notations, concepts and principles in the problem are defined.
	
	\item	A two-phase method, namely On-Shelf Utility Mining of Sequences (OSUMS), is designed with a novel storage data structure named periodical q-matrix. Two local pruning strategies are proposed to reduce the search space with two upper bounds time prefix extension utility (\textit{TPEU}) and time reduced sequence utility (\textit{TRSU}), and a strategy was designed to avoid redundant calculations.
	
	\item 	We also designed a one-phase method, namely OSUMS$^{+}$, which overcomes the two challenges present in the two-phase algorithms. OSUMS$^{+}$ adopts a novel storage data structure with periodical utility chain and utilizes two global pruning strategies with two upper bounds \textit{TPEU} and \textit{TRSU}.
	
	\item 	A series of experiments on six real-world datasets were conducted to evaluate the performance of the two proposed methods. The experimental results show that they outperform the only existing algorithm TP-HOUS. In addition, when compared to the one-phase algorithm OSUMS$^{+}$, OSUMS consumes a significantly large amount of memory and is not sufficiently efficient.
\end{itemize}

The remainder of this article is organized as follows. Section 2 briefly reviews related work about utility mining and OSUM. Section 3 presents basic definitions and formulates the problem of OSUMS in sequence data. In Section 4, we present the details of the two proposed methods with several novel data structures, upper bounds, and strategies. An experimental evaluation of the designed methods is presented in Section 5. Finally, conclusions and future work are discussed in Section 6.

%% file: 2_relatedwork.tex
\section{Related Work}

In this section, we separately review the prior literature pertaining to utility mining and OSUM. Particularly, we also discuss the drawbacks of the two-phase algorithms in the process of pattern mining, which also exists in the state-of-the-art algorithms in the domain of OSUM in a sequence database.

\subsection{Utility Mining}

The problem of HUIM is to select the high-utility itemsets (HUIs) from a transaction database where each row is an itemset and each item has a utility representing its importance based on the utility functions. Chan et al. \cite{chan2003mining} first incorporated the concept of utility for discovering highly desirable statistical patterns (i.e., HUIs); moreover, they presented a level-wise approach with a novel pruning strategy by relying on a weaker but anti-monotonic condition. Then, Liu et al. \cite{liu2005two} developed a milestone two-phase algorithm that performs two phases in the development of HUIM. They pioneered an upper bound named transaction-weighted utilization (\textit{TWU}), which satisfies the downward closure property. The anti-monotonicity of the upper bound guarantees that the complete set of HUIs is included in the candidates that the high TWU values generated in the first phase. After all candidates are identified, one more database scan is required to calculate their accurate utilities in the second phase, and the HUIs are outputted at the same time. Inspired by the two-phase algorithm, a series of two-phase methods such as incremental high-utility pattern \cite{ahmed2009efficient}, UP-Growth \cite{tseng2012efficient}, and Maximum Utility Growth \cite{yun2014high} were developed. When compared to the two-phase algorithm, the subsequent approaches extract HUIs performing two similar phases but avoid generating candidates that do not appear in the database in a pattern-growth manner. As explained, the two-phase algorithms have to retain a large number of candidates in memory and calculate their utilities in the second phase, which is the major causes of the high computational cost. In the worst situation discussed in \cite{fournier2019survey}, all candidates are evaluated in all transactions, which incurs several database scans. To avoid the scalability issue due to storing large-sized candidates, one-phase algorithms were designed, which achieved a significant breakthrough in terms of efficiency. In the mining process, the high utility of each candidate can be immediately identified once it is enumerated; thus, the memory for storing the candidate can be released soon. Till now, the one-phase HUIM approaches were extensively studied such as fast HUIM \cite{fournier2014fhm}, modified HUI-Miner \cite{peng2017mhuiminer}, and efficient HUIM \cite{zida2015efim} More details about HUIM can be obtained from comprehensive surveys \cite{fournier2019survey, gan2018surveyii}.


The primary goal of HUSPM is to identify the subsequences with high utility, which are referred to as high-utility sequential patterns (HUSPs), from a sequence database. Similar to the development history of HUIM, early HUSPM methods extracted desired HUSPs in two phases. Ahmed et al. are the pioneers who first incorporated utility into SPM \cite{ahmed2010novel}. They proposed a novel HUSPM framework for more real-world applicable information discovery. Further, they developed two two-phase algorithms, i.e., utility level (UL) and utility span (US). UL is a level-wise candidate generation method, which implies that it may generate candidates that do not appear in the database. The brute-force algorithm UL is more simple and straightforward than US; however, it generates too many candidates in the first phase and requires several database scans that are time consuming. The other algorithm, US, exploits a pattern-growth approach, which generates candidates based on relatively small-scale projected databases such that it requires a maximum of three database scans. With a striking resemblance to the two-phase HUIM algorithms, the two-phase methods of HUSPM also demonstrate two limitations. The first limitation is that a considerable amount of memory is occupied for storing the candidates generated in the first phase. The second one is that it is time consuming to compute actual utilities of candidates. To address the issue, a series of one-phase algorithms were proposed. Yin et al. \cite{yin2012uspan} formulated the HUSPM problem statement and designed a formal framework for HUSPM. A novel method USpan was introduced to discover the desired HUSPs using a new structure called lexicographic quantitative sequence tree (LQS-Tree), which represents the search space to be traversed. Moreover, they developed depth-first and width-first pruning strategies by adopting sequence-weighted utilization (\textit{SWU}) and sequence-projected utilization (\textit{SPU}) upper bounds, respectively. However, the \textit{SPU} is not a true overestimation of utility value, which implies that the width-first pruning strategy may prune the nodes representing HUSPs in the LQS-Tree. Then, Alkan et al. \cite{alkan2015crom} proposed the high utility sequential pattern extraction (HupsExt) algorithm based on a pattern tree. It can be noted that each node in the tree structure stores some necessary information for calculating the tighter upper bound named cumulated rest of match (\textit{CoRM}) and facilitating a pruning strategy called pruning before candidate generation. To more efficiently mine HUSPs, Wang et al. \cite{wang2016efficiently} extended USpan \cite{yin2012uspan} to HUS-Span, which can quickly discover patterns. Besides \textit{SWU}, the method adopts two tight upper bounds, prefix extension utility \textit{ (PEU)} and reduced sequence utility \textit{ (RSU)}, which can be easily obtained from a utility chain. The compact data structure expediates the calculation of not only upper bound values but also utility values. Till now, HUSPM algorithms have been considerably researched \cite{gan2020fast, gan2020proum}, and the mining process is becoming increasingly efficient owing to compact data structures and effective pruning strategies. 

\subsection{On-shelf Utility Mining}

Obviously, the aforementioned utility mining approaches by default consider that all items are on the shelf all the time in the market basket analysis. Thus, there exists a bias that the patterns having more on-shelf time are more likely to be extracted in HUSPM. To eliminate the bias, the problem of OSUM \cite{lan2011discovery, lan2014shelf, fournier2015foshu, radkar2015mining, dam2017efficient, lan2014discovery} was generalized in utility mining by considering the on-shelf time of items. Lan et al. \cite{lan2011discovery} first defined the task of OSUM in transaction data as the process of discovering the complete set of osHUIs from a transaction database with respect to a user-specified threshold. They also proposed an efficient two-phase algorithm with a data structure based on a periodical total transaction utility table that increases the execution efficiency. In the first phase, TP-HOU generates promising on-shelf utility itemsets that have high on-shelf utility in the current time period as candidates. Then, it calculates the actual on-shelf utility in the entire database in the second phase. As it is known, similar to the two-phase utility mining methods, the algorithm also suffers from the two drawbacks that degrade the efficiency. Consequently, their team extended the problem where certain items were associated with negative utility values \cite{lan2014shelf}. To cope with the issue, they designed the two-phase three-scan algorithm for mining osHUIs with negative profit (TS-HOUN), which requires only three database scans by adopting a proper upper bound and an effective itemset generation method. For further efficiency improvement, the mining process was reduced to one phase by Fournier et al. \cite{fournier2015foshu} in the faster osHUI mining (FOSHU) algorithm. FOSHU handles all time periods simultaneously and introduces novel strategies to handle negative values efficiently; consequently, it runs 1,000 times faster than the state-of-the-art TS-HOUN. Moreover, certain extension problem were generalized to extract some interesting patterns; for example, discovering top-$k$ osHUIs \cite{dam2017efficient} and mining osHUIs from dynamic updated databases \cite{radkar2015mining}.

Considering the on-shelf time, utility values, and sequential order, Lan et al. \cite{lan2014discovery} presented a new research issue, that is OSUM of sequence data, which is intrinsically more complex than the task in transaction data owing to a combinatorial explosion of the huge search space. They also developed a two-phase method TP-HOUS to efficiently mine osHUSPs in a temporal quantitative sequence database. Moreover, the sequence-utility upper bound (\textit{SUUB}), as well as a corresponding pruning strategy were designed to speed up the mining process. As a two-phase algorithm, TP-HOUS suffers from the common limitations of the aforementioned two-phase methods. To the best of our knowledge, TP-HOUS is the only existing method used for OSUM of sequence data. There is a significant room for improvement in terms of execution time, memory consumption, and scalability.

%% file: 3_preliminaries.tex
\section{Definitions and Problem Statement}

In this subsection, we introduce the significant definitions, including concepts, notations, and principles, used in the domain of OSUM. It is to be noted that certain basic definitions are derived from prior works \cite{lan2014discovery,gan2020proum}. Moreover, we formulated a normative problem statement of OSUM of sequence data.

\subsection{Preliminaries}

Let the finite set $I$ = $\{i_{1},i_{2},\ldots,i_{n}\}$ be a universal set, where the elements are the items that may appear in the problem. An itemset $v$ is a set of distinct items and satisfies $v \subseteq I$. If the length of $v$, that is, the number of items contained in the itemset, is equal to $l$, the itemset $v$ is called an $l$-itemset. A sequence $s$ = $<$$v_{1},v_{2},\ldots,v_{m}$$>$ is a list of itemsets in chronological order, where $v_{k} \subseteq I$ for $1 \leq k \leq m$. The number of elements contained in $s$ is the size of $s$. We define $|s|$ = $\sum_{k = 1}^{m}|v_{k}|$ as the length of $s$, and $s$ is called an $l$-sequence if $|s|$ = $l$. For example, $v$ = \{\textit{a} \textit{b}\} is a two-itemset, and $s$ = $<$\{\textit{a} \textit{f}\}, \{\textit{b} \textit{e} \textit{f}\}\}$>$ is a five-sequence as it consists of five items. 

\begin{definition}
	A quantitative item is defined as an ordered tuple ($i$:$q$) where $i \in I$ and $q$ is a positive integer representing the internal utility of item $i$. A quantitative itemset, denoted as $w$ = \{(\textit{$i_{1}$}:$q_{1}$) (\textit{$i_{2}$}:$q_{2}$) $\ldots$ (\textit{$i_{a}$}:$q_{a}$)\}, is a set of $q$-items. Similarly, we also quantify the sequence and define the quantitative sequence $r$ = $<$$w_{1},w_{2},\ldots,w_{b}$$>$, where $w_{k}$ is a quantitative itemset for $1 \leq k \leq b$. For brevity, we use the prefix symbol "$q$-" to denote the term quantitative. For example, a quantitative item/quantitative itemset/quantitative sequence can be denoted as $q$-item/$q$-itemset/$q$-sequence, respectively. 
\end{definition}

\begin{definition}
	Let $T$ = $\{t_{1},t_{2},\ldots,t_{n}\}$ denote a set of mutually disjoint time periods. A temporal $q$-sequence database $D$ is a set of $q$-sequences, i.e., $D$ = $\{\textit{QS}_{1,1},\textit{QS}_{1,2},\ldots,\textit{QS}_{n,m}\}$, where $\textit{QS}_{i,j}$ is $j$-th $q$-sequence in $i$-th time period. In general, a temporal $q$-sequence database $D$ can be represented as a set of triples, each of which is denoted as (\textit{TID}, \textit{SID}, $s$) where $s$ is a $q$-sequence, $\textit{TID} \in T$ is the time period of $s$ occurring, and \textit{SID} is the unique identifier of $s$ in the time period $t$. In addition, each item $i_{j} \in I$ that appears in the temporal $q$-sequence database is associated with a positive integer named external utility and its on-shelf time period information. To facilitate the statement, all $q$-sequences within the $t_j$ time period in $D$ can be identified as $\textit{TD}_{t_j}$.
\end{definition}

For convenience, we assume all items/$q$-items in an itemset/$q$-itemset are arranged in alphabetical ascending order in the remainder of this paper. As a running example, a temporal $q$-sequence database, including five $q$-sequences and six types of items, is listed in Table \ref{table1}, where $<$\{(\textit{b}:1) (\textit{d}:3)\}, \{(\textit{c}:4) (\textit{e}:1)\}$>$ is the first $q$-sequence within the time period \(\textit{t}_{1}\). Each item listed in Table \ref{table1} is associated with an external utility listed in Table \ref{table2}. In addition, the on-shelf time information of items is listed in Table \ref{table3}, where 1 represents the item is on the shelf in the current time period. For example, \textit{c} is on the shelf all the time, while \textit{f} is only on the shelf in \(\textit{t}_{3}\). In addition, $\textit{QS}_{1,1}$ and $\textit{QS}_{1,2}$ are included in $\textit{TD}_{1}$.
 
\begin{table}[!t]
	\centering
	\caption{Running example of a temporal $q$-sequence database}
	\label{table1}
	\begin{tabular}{|c|c|c|}  
		\hline 
		\textbf{TID} & \textbf{SID} & \textbf{$q$-sequence} \\
		\hline  
		\(\textit{t}_{1}\) & \(\textit{S}_{1}\) & $<$\{(\textit{b}:1) (\textit{d}:3)\}, \{(\textit{c}:4) (\textit{e}:1)\}$>$ \\ 
		\hline
		\(\textit{t}_{1}\) & \(\textit{S}_{2}\) & $<$\{(\textit{b}:2) (\textit{e}:3)\}, \{(\textit{c}:4)\}, \{(\textit{b}:1) (\textit{c}:3)\}$>$ \\  
		\hline  
		\(\textit{t}_{2}\) & \(\textit{S}_{1}\) & $<$\{(\textit{c}:3) (\textit{d}:4)\}, \{(\textit{a}:3) (\textit{c}:1)\}, \{(\textit{a}:2) (\textit{c}:3) (\textit{d}:1)\}$>$ \\
		\hline  
		\(\textit{t}_{2}\) & \(\textit{S}_{2}\) & $<$\{(\textit{a}:3) (\textit{d}:2)\}, \{(\textit{a}:1) (\textit{e}:2)\}, \{(\textit{c}:3)\}, \{(\textit{b}:2) (\textit{c}:4)\}$>$ \\
		\hline
		\(\textit{t}_{3}\) & \(\textit{S}_{1}\) & $<$\{(\textit{a}:4) (\textit{e}:2) (\textit{f}:2)\}, \{(\textit{c}:1) (\textit{e}:3)\}$>$ \\
		\hline
	\end{tabular}
\end{table}

\begin{table}[!t]
	\caption{Utility table}
	\label{table2}
	\centering
	\begin{tabular}{|c|c|c|c|c|c|c|}
		\hline
		\textbf{Item}	    & \textit{a}	& \textit{b}	& \textit{c}	& \textit{d}	& \textit{e}	& \textit{f} \\ \hline 
		\textbf{External utility}	& \$2 & \$3& \$1 & \$1 & \$2 & \$4 \\ \hline
	\end{tabular}
\end{table}

\begin{table}[!t]
	\caption{On-shelf time periods of items}
	\label{table3}
	\centering
	\begin{tabular}{|c|c|c|c|}
		\hline
		\diagbox{\textbf{Item}}{\textbf{Time}} & $t_1$	& $t_2$	& $t_3$ \\
		\hline 
		\textit{a} & 0 & 1 & 1 \\
		\hline 
		\textit{b} & 1 & 0 & 0 \\
		\hline 
		\textit{c} & 1 & 1 & 1 \\
		\hline 
		\textit{d} & 1 & 1 & 0 \\
		\hline 
		\textit{e} & 1 & 1 & 1 \\
		\hline 
		\textit{f} & 0 & 0 & 1 \\
		\hline
	\end{tabular}
\end{table}

Let us consider a $q$-itemset $w$ = \{(\textit{$i_{1}$}:$q_{1}$) (\textit{$i_{2}$}:$q_{2}$) $\ldots$ (\textit{$i_{a}$}:$q_{m}$)\} with length of $m$. We define the utility of the item $i_{k}$ within $w$ as $u(i_{k},w)$ = $q(i_{k},w) \times p(i_{k})$, where $q(i_{k},X)$ is the internal utility (usually representing quantity) of the item $i_{k}$ within $w$, and $p(i_{k})$ (usually representing unit profit) is the external utility of $i_{k}$ for $1 \leq k \leq m$. For the $q$-itemset $w$, its utility can be defined as $u(w)$ = $\sum_{k = 1}^{m}u(i_{k},w)$. Then, the utility of a $q$-sequence $s$ = $<$$v_{1}$, $v_{2}$, $\ldots$, $v_{n}$$>$, denoted as $u(s)$, is defined as $u(w)$ = $\sum_{k = 1}^{n}u(v_k)$. Moreover, the utility of a temporal $q$-sequence database $D$, denoted as $u(D)$, is the sum of the utility of each $q$-sequence contained in $D$. 

Consider the example in Table \ref{table1}. The utility of the item \textit{f} within the first itemset in $\textit{QS}_{3,1}$ can be calculated as $u$(\textit{b},\{(\textit{a}:4) (\textit{e}:2) (\textit{f}:2)\}) = 2 $\times$ \$4 = \$8. Then, we have $u(\textit{QS}_{3,1})$ = \$20 + \$7 = \$27 and $u(D)$ = \$12 + \$22 + \$22 + \$27 + \$27 = \$110.

\begin{definition}
	\label{contain}
	Let there be two itemsets $v$ and $w$. Let us consider that $v$ is a subset of $w$, denoted as $v \subseteq w$, if $\forall i \in v$ satisfies $i \in w$. Moreover, given two sequences $s$ = $<$$w_{1},w_{2},\ldots,w_{m}$$>$ and $r$ = $<$$v_{1},v_{2},\ldots,v_{n}$$>$, $s$ is a subsequence of $r$, denoted as $s \subseteq r$, if and only if there exists $m$ integers $ 1 \leq b_{1} < b_{2} < \ldots < b_{m} \leq n$ such that $w_{p}\subseteq v_{b_{p}}$ for $ 1 \leq p \leq m$.
\end{definition}

\begin{definition}
	\label{match}
	Given a $q$-sequence $s$ = $<$$v_1$, $v_2$, $\ldots$, $v_d$$>$ and sequence $r$ = $<$$w_1$, $w_2$, $\ldots$, $w_d'$$>$, if $d$ = $d'$ and the items in $v_k$ are the same as those in $w_k$ for $1 \leq k \leq d$, then we say r matches s, which is denoted as $r \sim s$. 
\end{definition}

\begin{definition}
	Given a $q$-sequence $s$ and a sequence $r$, if $s \sim s' \land r \subseteq s'$, then we say $r$ is contained in $s$. In the remainder of this article, for convenience, $r \sim s$ is used to indicate that $r$ is contained in $s$ ($s \sim s' \land r \subseteq s'$).
\end{definition}

For instance, $v_1$ = \{\textit{a}\} and $v_{2}$ = \{\textit{a} \textit{c}\} are both subsets of $w$ = \{\textit{a} \textit{b} \textit{c}\}, and $<$\{\textit{b}\}, \{\textit{c}\}$>$ is a subsequence of the sequence $<$\{\textit{b} \textit{d}\}, {\textit{c} \textit{e}\}$>$. Consider the running example where the sequence $<$\{\textit{b} \textit{d}\}, \{\textit{c} \textit{e}\}$>$ matches the $q$-sequence $\textit{QS}_{1,1}$. Then, we can say that the sequence $<$\{\textit{b}\}, \{\textit{c}\}$>$ is contained in $\textit{QS}_{1,1}$. Note a special circumstance where a sequence matching the $q$-sequence is also contained in it.

\begin{definition}
	Let us consider a sequence $r$ = $<$$w_1$, $w_2$, $\ldots$, $w_m$$>$ is contained in a $q$-sequence $s$ = $<$$v_1$, $v_2$, $\dots$, $v_n$$>$; according to Definition \ref{contain}, we assume that the integer sequence is $ 1 \leq k_{1} \leq k_{2} \ldots \leq k_{m} \leq n$; then, we say $r$ has an instance in $s$ at position $p$: $<$$k_{1}$, $k_{2}$, ..., $k_{m}$$>$.
\end{definition}

Note a sequence could have multiple instances in a $q$-sequence at different positions, and different instances may have the same extension position. Consider the sequence $r$ = $<$\{\textit{a}\}, \{\textit{c}\}$>$ in Table \ref{table1}. It has four instances in $\textit{QS}_{2,2}$ at positions $<$1, 3$>$, $<$2, 3$>$, $<$1, 4$>$, and $<$2, 4$>$, where the extension position of the first two instances is three.

\begin{definition}
	Let there be a sequence $r$ = $<$$w_1$, $w_2$, $\ldots$, $w_m$$>$ and a $q$-sequence $s$ = $<$$v_1$, $v_2$, $\ldots$, $v_n$$>$. Let us suppose $r$ has an instance in $s$ at position $p$: $<$$k_{1}$, $k_{2}$, ..., $k_{m}$$>$. Then, the utility of the instance is denoted as $u(r,\ p,\ s)$ and defined as $u(r,\ p,\ s)$ = $\sum_{\forall i \in r}^{}{u(i,\ k_j,\ s)}$, where $k_j$ is the corresponding position of item $i$ in $p$.
\end{definition}

For example, in Table \ref{table1}, the first instance of $r$ = $<$\{\textit{a}\}, \{\textit{c}\}$>$ in $\textit{QS}_{2,2}$ is at position $p$ = $<$1, 3$>$; then, the utility of this instance can be calculated as $u(r,\ p,\ s)$ = 3 $\times$ \$2 + 3 $\times $ \$1 = \$9. 

\begin{definition}
	Let there be a sequence $r$ and $q$-sequence $s$. Assume that the set of extension positions of $r$ in $s$ is $P$ = $\{p_1,p_2,\ldots,p_n\}$. The utility of $r$ in $s$ at extension position $p_i$ is defined as $u(r,\ p_{i},\ s)$ = $\max$\{$u$(t, $<$$ j_{1}$, $j_{2}$, $\ldots$, $p_{i}$$>$, s)$\ |\ j_{1} \leq j_{2} \leq \ldots \leq p_{i}$ and $<$$ j_{1},j_{2},\ldots,p_{i}$$> \in P(t,\ p_{i},\ s)$\}, where $P(r,\ p_{i},\ s)$ is the set of positions with extension position \(p_{i}\). The utility of $r$ in $s$ is the max utility value of all utilities at all its extension positions, which is defined as $u(r,s)$ = $\max\{u(r,\ p,\ s)\ |\ p \in P\}$.
\end{definition}

For example, in Table \ref{table1}, the first two instances of the sequence $r$ = $<$\{\textit{a}\}, \{\textit{c}\}$>$ in $\textit{QS}_{2,2}$ is at positions $<$1,3$>$ and $<$2,3$>$ with the same extension position 3. Then, we have $u(r,\ p_{i},\ \textit{QS}_{2,2})$ = $\max\{\$9, \$5\}$ = \$9. The utility of $r$ in $\textit{QS}_{2,2}$ is $u(r,\textit{QS}_{2,2})$ = $\max\{\$9, \$10\}$ = \$10.

\begin{definition}
	The sequence utility of a $q$-sequence $s$, denoted as $\textit{su}(s)$, is the sum of the utilities of all items contained in $s$. The periodical total sequence utility of time period $t$, denoted as $\textit{ptsu}(t)$, is defined as the sum of the sequence utilities of all $q$-sequences within the time period $t$ \cite{lan2014discovery}. 
\end{definition}

\begin{definition}
	 Given a sequence $r$ and time period $t$, the periodical utility of $r$ within $t$ is the sum of the utilities in each $q$-sequence with \textit{TID} = $t$, which is defined as $\textit{pu}(r,t)$ = $\sum_{r \subseteq s \land s \in \textit{TD}_{t}}u(r,s)$.
\end{definition}

Let us consider the time period $t_1$ in Table \ref{table1}. The sequence utility of $\textit{QS}_{1,1}$ is $\textit{su}(\textit{QS}_{1,1})$ = \$3 + \$3 + \$4 + \$2 = \$12; further, we have $\textit{ptsu}(t_1)$ = \$12 +\$22 = \$34. For example, in Table \ref{table1}, the periodical utility of $r$ = $<$\{\textit{a}\}, \{\textit{c}\}$>$ within time period 2 is $\textit{pu}(r,t_2)$ = \$9 + \$10 = \$19.

\begin{definition}
	The set of on-shelf time periods of a sequence $r$, denoted as $\textit{ot}(r)$, is the union of on-shelf time periods of all items contained in $r$. Given a temporal $q$-sequence database, the on-shelf utility of a sequence $r$, denoted as $\textit{ou}(r)$, is defined as $\textit{ou}(r)$ = $\sum_{s \in \textit{TD}_{t} \land t \in \textit{ot}(r)}u(r,s)$. Moreover, the on-shelf utility ratio of $r$ is defined as $\textit{our}(r)$ = $\textit{ou}(r)/\sum_{t \in\textit{ot}(r)}\textit{ptsu}(t)$.
\end{definition}

For example, given the sequence $r$ = $<$\{\textit{a}\}, \{\textit{c}\}$>$, we have $\textit{ot}(r)$ = $\{t_2,t_3\}$. Let us consider the running example in Table \ref{table1}. Suppose $r$ = $<$\{\textit{a}\}, \{\textit{c}\}$>$, we have $\textit{ou}(r)$ = \$19 + \$9 = \$28, as $r$ is only on the shelf within time periods $t_2$ and $t_3$. Then, we can easily obtain $\textit{our}(r)$ = \$28 / (\$49 + \$27) = 36.8\%. 

%% file: 4_technic.tex
\section{Proposed MDUS$ _\text{EM} $ Algorithm} 

By referring to the only existing algorithm TP-HOUS \cite{lan2014discovery} that splits the mining process into two phases, we designed an efficient two-phase method, namely OSUMS. Similar to TP-HOUS, OSUMS scans the original database once to discover all promising osHUSPs in the first phase. Then, in the second phase, OSUMS calculates the actual on-shelf utility values of the items and output the complete set of osHUSPs. For further efficiency improvement, we introduce several novel data structures, upper bounds, and strategies. Although OSUMS is more efficient than TP-HOUS, it also suffers from the intrinsic two limitations of the two-phase approaches. The first limitation is that retaining the promising osHUSPs generated in the first phase occupy large amounts of memory; the second limitation is that identifying whether each promising osHUSP has an on-shelf high utility (i.e., calculating their actual on-shelf utility) incurs significantly high computational costs. To overcome these issues, we improved OSUMS and proposed a one-phase method OSUMS$^{+}$. When compared to OSUMS, the OSUMS$^{+}$ method utilizes similar storage data structures and upper bounds, but adopts two global strategies to check one promising osHUSP immediately as long as it is extracted. 

Note that if a sequence is not a promising osHUSP within each time period, then it is absolutely impossible for it to be an osHUSP in the database. We present the details of the two approaches below.

\subsection{OSUMS Approach}
\label{osums}

To discover osHUSPs efficiently, TP-HOUS adopts an upper bound \textit{SUUB} as well as a pruning strategy to reduce the search space. However, the upper bound \textit{SUUB} is so loose that TP-HOUS generates abundant candidates in the mining process. To significantly reduce the number of candidates, OSUMS adopts two upper bounds \textit{TPEU} and \textit{TRSU} based on \textit{PEU} and \textit{RSU} designed in \cite{wang2016efficiently}, respectively, as well as two local pruning strategies. Note that the processes of determining promising osHUSPs in OSUMS within each time period are independent; thus, we called the pruning strategies in OSUMS local ones. We also designed two storage data structures that are convenient for calculating the two upper-bound values and the actual on-shelf utility values. In addition, with the purpose of expediating the checking process that determines osHUSPs in the second mining process, we developed a novel structure named candidate tree (\textit{CTree}) and an efficient strategy. In theory, OSUMS is able to outperform the relatively simple method TP-HOUS.

To facilitate the statement of the proposed methods, first, we present certain basic and essential definitions as shown below.

\begin{definition}
	Let there be a sequence $r$ = $<$$w_1$, $w_2$, $\ldots$, $w_m$$>$ and an item $i$. We define the $I$-Extension operation of $r$ as the process that appends $i$ to the last itemset $w_m$. The operation results in a new sequence $r'$, which is denoted as $<$$ t\oplus i $$>$ and is called an $I$-Extension sequence of $r$. The other operation of $S$-Extension of $r$ is defined as placing a new itemset $w_{m+1}$ that only contains $i$ behind $w_m$, which also generates a new sequence $r'$. Generated by the $S$-Extension operation, $r'$, denoted as $<$$t\otimes i$$>$, is called an $S$-Extension sequence of $r$.
\end{definition}

Consider an example $r$ = $<$\{\textit{a}\}$>$, $r_1$ = $<$\{\textit{a} \textit{c}\}$>$ is an $I$-Extension sequence of $r$, while $r_2$ = $<$\{\textit{a}\}, \{\textit{c}\}$>$ is an $S$-Extension sequence of $r$. Obviously, any nonempty sequence can be generated by a series of $I$/$S$-Extension operations from an empty sequence $<$ $>$.

\begin{definition}
	Given a sequence $r$ and $q$-sequence $s$, let us assume that $r$ has instances in $s$ at extension position $p$ and the extension item is $i'$. The rest utility $r$ in $s$ at extension position $p$, denoted as $\textit{ru}(r,\ p,\ s)$, is defined as $\textit{ru}(r,\ p,\ s)$ = $\sum_{i \in s \land i \succ i'}^{}u(i,\ j',\ s)$, where $p \geq j'$ and $i$ represents the $q$-items located behind $i'$ in $s$. 
\end{definition}

Let us consider the example listed in Table \ref{table1}. The sequence $r$ = $<$\{\textit{a}\}, \{\textit{a}\}$>$ has an instance in $\textit{QS}_{2,2}$ at the extension position 2; then, we obtain $\textit{ru}(r,\ 3,\ \textit{QS}_{2,2})$ = \$4 + \$3 + \$6 +\$4 = \$17.

Let us consider a sequence $r$ and $q$-sequence $s$ and let us assume that $r$ has an instance in $s$ at extension position $p$. According to \cite{wang2016efficiently}, the \textit{PEU} upper bound of $r$ in $s$ at $p$ is defined as: 
\[\textit{PEU}(r,\ p,\ s)\begin{cases}
u(r,\ p,\ s)+\textit{ru}(r,\ p,\ s) & \textit{ru}(r,\ p,\ s)>0\\
0&otherwise
\end{cases}\]
Suppose $r$ has several instances in $s$ at $m$ extension positions, $p$ = $\{p_1,p_2,\ldots,p_m\}$. The \textit{PEU} value of $r$ in $s$ is the maximum value of \textit{PEU} at all extension positions, which can be represented as: $\textit{PEU}(r,s)$ = $\max\{ \textit{PEU}(r,\ p_{1}$, $ s)$, $\ldots$, $\textit{PEU}(r,\ p_{m}, s)\}$. Based on \textit{PEU}, we introduce the \textit{TPEU} upper bound, which can be defined as $\textit{TPEU}(r, t)$ = $\sum_{s \in \textit{TD}_{t} \land r \subseteq s}^{}\textit{PEU(r,s)}$.

For instance, in Table \ref{table1}, let us consider a sequence $r$ = $<$\{\textit{a}\}, \{\textit{c}\}$>$. We obtain the \textit{PEU} values $\textit{PEU}(r,\ \textit{QS}_{2,1})$ = \$9 + \$1 = \$10 and $\textit{PEU}(r,\ \textit{QS}_{2,2})$ = \$9 + \$10 = \$19. Thus, the \textit{TPEU} of $r$ in time period $t_2$ can be calculated as $\textit{TPEU}(r, t_2)$ = \$10 + \$19 = \$29.

Assume the sequence $r$ is an extension sequence of the sequence $\alpha$, then the \textit{RSU} upper bound of $r$ in a $q$-sequence $s$ \cite{wang2016efficiently} is defined as follows. 
\[\textit{RSU}(r,s)\begin{cases}
\textit{PEU}(\alpha,s) & \alpha \subseteq s \land r \subseteq s\\
0&otherwise
\end{cases}\]
Based on \textit{RSU}, we define a novel upper bound \textit{TRSU} that can be adopted in the problem of OSUM in sequence data, that is, $\textit{TRSU}(r,t)$ = $\sum_{s \in D \land {TD}_{t} \subseteq s}^{}\textit{RSU(r,s)}$.

Let there be two sequences $r$ = $<$\{\textit{a}\}, \{\textit{c}\}$>$ and $r'$ = $<$\{\textit{a}\}, \{\textit{c}\}, \{\textit{b}\}$>$. It can be determined that $\textit{TRSU}(r',t_2)$ = \$19 as $r' \subseteq \textit{QS}_{2,2}$, $r' \subseteq \textit{QS}_{2,2}$, and $\textit{PEU}(r,\ \textit{QS}_{2,2})$ = \$19, which we have calculated.

To clearly explain the proposed algorithms, we design a structure called periodical lexicographic sequence forest (\textit{PLS-Forest}) to represent the entire search space in the mining process of OSUMS based on the concept of lexicographic sequence tree (\textit{LS-Tree}) \cite{yin2013efficiently}. \textit{LS-Tree} is a tree structure where each node represents a sequence, with the root representing an empty sequence. For a node in the \textit{LS-Tree}, the sequence represented by it is the extension sequence of that of its parent. Each \textit{LS-Tree} in \textit{PLS-Forest} represents the search space within one of the time periods. We show an example of \textit{PLS-Forest} in Figure \ref{forest}. The \textit{LS-Trees} in a \textit{PLS-Forest} may be different as items are on the shelf at different time periods. Without loss of generality, we present all children of a node in alphabetical ascending order. Note that \textit{PLS-Forest} and \textit{LS-Tree} are both abstractly conceptual structures, and the real search space may be different depending on different situations \cite{gan2020proum}. 

\begin{figure}[htbp]
	\centering
	\includegraphics[width=0.7\linewidth]{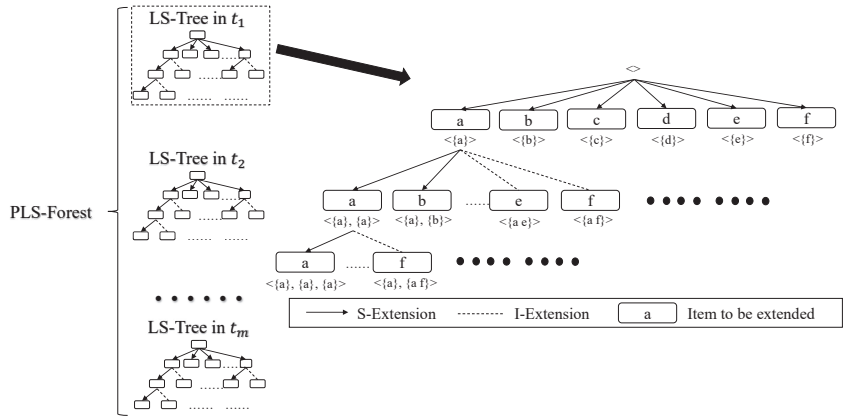}
	\caption{Partial periodical lexicographic sequence forest for the running example}
	\label{forest}
\end{figure}

\subsubsection{Storage Data Structure}

Based on the data structure $q$-matrix in USpan \cite{yin2012uspan}, we introduce a novel data structure, namely periodical $q$-matrix, where each matrix can represent a $q$-sequence in a temporal $q$-sequence database. Besides including a matrix for storing utility and rest utility information, a periodical $q$-matrix also consists of the time period and identifier of the $q$-sequence. In practice, the periodical $q$-matrices with the same time period are placed in one list, which can be indexed by the time period in memory. For better visualization, we present the periodical $q$-matrices of the running example in Figure \ref{matrix}, where we only show details of the periodical $q$-matrix of $\textit{QS}_{2,1}$ for brevity. Here, we briefly introduce the data structure of the $q$-matrix, where elements can be indexed by on-shelf items within the time period and $q$-itemset numbers. Each element has three values, where the first value shows the utility of the $q$-item and second is the sum of utilities of $q$-items behind it (also called rest utility). The items that do not appear in the $q$-itemset are given a utility value of zero.

Let us observe the record for item \textit{a} in the periodical $q$-matrix of $\textit{QS}_{2,1}$ in Figure \ref{matrix}. The terms in the first entry are both \$0, as \textit{a} does not appear within the first $q$-itemset in $\textit{QS}_{2,1}$. According to the definition of utility, $u(a,\ 2,\ \textit{QS}_{2,1})$ = 3 × \$2 = \$6 and $\textit{ru}$ = \$1 + \$4 + \$3 + \$1 = \$9 can be calculated; then, we obtain the second entry (\$6, \$9). The remaining calculations can be performed in the same manner.

\begin{figure}[htbp]
	\centering
	\includegraphics[width=0.7\linewidth]{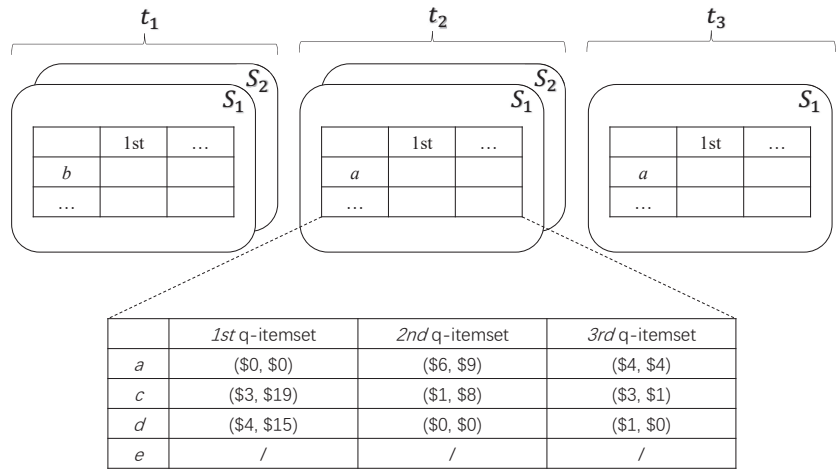}
	\caption{Periodical $q$-matrices for the running example}
	\label{matrix}
\end{figure}

A simple and intuitive method to calculate the utility of a candidate is by scanning the entire database. Obviously, the brute-force method incurs significantly high computational costs as the method has to check the $q$-sequences that have no possibility to contain the candidate. It is easy to understand that a $q$-sequence may contain a sequence only if it contains the prefix of the sequence. Thus, OSUMS recursively constructs projected databases represented by a compact and efficient data structure utility chain \cite{wang2016efficiently} for reducing the scan scope. The projected databases of candidates have a relatively small scale and store the necessary information for calculating the values of utilities and upper bounds. It is noted that the candidate refers to those sequences that must be checked to determine if they have high on-shelf utility. The utility chain of a sequence $r$ consists of multiple utility lists and a head table. The head table includes a series of tuples (\textit{SID}, \textit{PEU}), each of which corresponds to a $q$-sequence containing $r$ and indexes a utility list. The \textit{SID} value is the identifier of the $q$-sequence within the time period, and \textit{PEU} is the upper bound value of $r$ in this $q$-sequence. Suppose $r$ has $m$ extension positions $P$: $\{p_1$, $p_2$, $\ldots$, $p_m\}$ in the $q$-sequence $s$, the corresponding utility list of $s$ has $m$ utility elements that consist of the following three fields: 1) field \textit{tid} presents the $i$-th extension position $p_i$, 2) field \textit{acu} is the utility of $r$ at the $i$-th extension position $p_i$ (i.e., $u(r,\ p_i,\ s)$), and 3)field \textit{ru} shows the rest utility of $r$ at the $i$-th extension position $p_i$ in $s$ (i.e., $\textit{ru}(\textit{fs},\ p_i,\ S)$). 

\begin{figure}[htbp]
	\centering
	\includegraphics[width=0.56\linewidth]{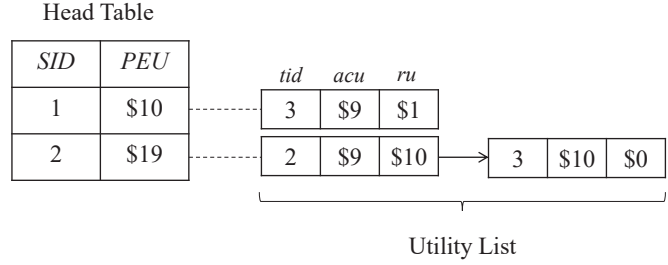}
	\caption{Projected database of $r$ = $<$\{\textit{a}\}, \{\textit{c}\}$>$ in OSUMS}
	\label{db}
\end{figure}

Figure \ref{db} illustrates the projected database, represented by a utility chain, of the sequence $r$ = $<$\{\textit{a}\}, \{\textit{c}\}$>$ within the time period $t_2$. Let us consider the second utility list, which corresponds to $\textit{QS}_{2,2}$, as an example. $r$ has two extension positions in $\textit{QS}_{2,2}$, that is, 2 and 3; therefore, the utility list consists of two elements. The utilities as well as the rest utilities at these two extension positions are $u(r,\ 2,\ s)$ = \$9, $u(r,\ 3,\ s)$ = \$10, $\textit{re}(r,\ 2,\ s)$ = \$10, $\textit{re}(r,\ 3,\ s)$ = \$0, respectively. Thus, the \textit{PEU} value can be \$19 according to the definition of \textit{PEU}. As it can be observed, the utility chain consists of essential and compact information of $r$; consequently, scanning the projected database is sufficient to check whether a candidate has high on-shelf utility. Based on the projected database of an $l$-sequence, the projected databases of its extension sequences (i.e., $l+1$-sequence) can be constructed. More details about utility chain can be referred to \cite{wang2016efficiently}.

To avoid the redundant utility calculations in the second mining process, we propose a novel data structure called \textit{CTree} based on a trie (i.e., prefix tree) for storing certain calculation results in the first mining process. Just as its name implies, the tree structure buffers the candidates whose periodical utility values have been calculated within at least one time period. Moreover, the structure also retains the related information by indicating whether it must be checked at certain time periods. All the descendants of a node have a common prefix of the string associated with that node, and the root is associated with the empty sequence. There are certain virtual nodes, denoted as -1, which indicate the end of the itemset. Excluding these virtual nodes, any other node represents a sequence. To facilitate efficient utility calculation and support multiple database updates, each node $N$ in \textit{CTree} contains the following auxiliary information:
\begin{itemize}
	\item \textit{Item}: The item of $N$.
	\item \textit{Seq}: The candidate represented by $N$.
	\item \textit{On-shelfTime}: A bit set represents the on-shelf information of \textit{Seq}. The size is equal to the number of time periods. The $i$-th bit value is set to 1 when \textit{Seq} is on the shelf in the $i$-th time period. Otherwise, this value is set to 0.
	\item \textit{CalculatedTime}: A bit set represents the time periods when the utility of \textit{Seq} has been calculated. The size is also equal to the number of time periods. The $i$-th bit value is set to 1 when the utility of \textit{Seq} has been calculated in the $i$-th time period. Otherwise, this value is set to 0.
	\item \textit{ChildrenList}: A list containing the children of node $N$.
	\item \textit{isPromising}: A Boolean value indicating whether \textit{Seq} is a promising osHUSP within at least one time period. 
	\item \textit{Utility}: A list that stores the periodical utility values of \textit{Seq} within time periods. The size of the list is the number of time periods. If \textit{Seq} is not on the shelf in $i$-th time period, or it is on the shelf but its periodical utility in the $i$-th time period has not been calculated, the bit value is set to 0. Otherwise, this value is set to 1.
	\item \textit{CUtility}: The sum of values in the list \textit{Utility}.
\end{itemize}

\begin{figure}[htbp]
	\centering
	\includegraphics[width=0.7\linewidth]{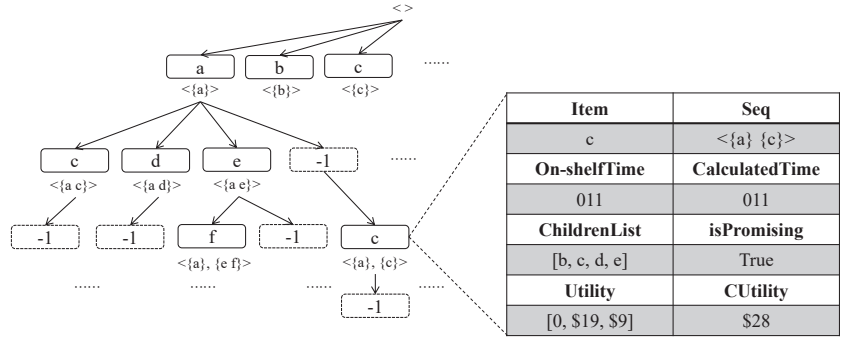}
	\caption{Partial candidate tree for the running example}
	\label{ctree}
\end{figure}

Figure \ref{ctree} illustrates the \textit{CTree} of the running example listed in Table \ref{table1} when $\xi$ = 0.3. For brevity, we only provide the values of \textit{Item} and \textit{Seq} for each node, and the full structure of the node $N$ that represents the sequence $r$ = $<$\{\textit{a}\}, \{\textit{c}\}$>$. Subsequently, we discuss the process of constructing the node $N$ and how to update its information when $\xi$ = 0.3. In the first phase, $r$ is generated as a candidate and its periodical utility is calculated as \$19 within the time period $t_2$. Then, OSUMS builds the node $N$ representing $<$\{\textit{a}\}, \{\textit{c}\}$>$ as the child of the node $N'$ representing $<$\{\textit{a}\}$>$ because the candidate $r$ first occurs in the mining process. After construction, OSUMS updates the information of $N$ and the children of $N'$. Finally, OSUMS sets the Boolean field \textit{isPromising} to TRUE, as $r$ is a promising osHUSP within the time period $t_2$. Next, when OSUMS generates $r$ again within the time period $t_3$, updating the auxiliary information of $N$ and $N'$ is sufficient without a construction operation.

In the second phase, OSUMS traverses the \textit{CTree} with a depth-first strategy. If the sequence represented by the current node is not a promising osHUSP within any time period, then it is skipped. If not, the AND operation is used to obtain the time periods when the candidate is on the shelf but its periodical utility has not been calculated. Then, OSUMS checks whether the candidate is an osHUSP by a series of calculations. Clearly, OSUMS fully exploits the calculation results of the first mining process and avoids the redundant utility calculations. Furthermore, let us assume that the average length of candidates is $m$ and number of candidates is $n$; then, maintaining the \textit{CTree} with time complexity \textit{CTree} $O(m)$ is a linear time operation; however, updating the information of promising osHUSPs is of polynomial time $O(mn)$ in TP-HOUS \cite{lan2014discovery}.

\subsubsection{Strategies}

As it can be observed in Figure \ref{forest}, \textit{LS-Forest} is an enormous structure, especially when there are a large number of items appearing in the database. To avoid critical combinatorial explosion of the search space, we propose two pruning strategies according to the download closure property of the two upper bounds \textit{TPEU} and \textit{TRSU}. The two upper bounds are tighter than the only existing one \textit{SUUB} \cite{lan2014discovery} in the problem of OSUM of sequence data, which implies that our method OSUMS is able to reduce more search space and generate less candidates. It is not difficult to understand that the CTree may be different under various $\xi$ settings. This is because certain branches in \textit{PLS-Forest} could be pruned by the pruning strategies, and the representing candidates are not generated.

\begin{theorem}
	\label{TPEU}
	Given a temporal $q$-sequence database $D$ and two sequences $r$ and $r'$, assuming that the node representing $r$ is a descendant of the node representing $r'$ (i.e., $r'$ is a prefix sequence of $r$), it can be obtained within a time period $t$ in $D$ that $$\textit{our}(r,t) \leq \textit{TPEU}(r',t)/\textit{ptsu}(t).$$
\end{theorem}

\begin{proof}
	\label{proof_TPEU}
	Let $r$ = $r' \bullet r''$, where $\bullet$ is a common concatenation operation of sequences. $r''$ is a nonempty sequence as $r'$ is a prefix sequence of $r$. Given a $q$-sequence $s \in \textit{TD}_t$, if $r$ is contained in $s$, then $r'' \subseteq s$ and $r' \subseteq s$. The utility of $r$ in $s$ is the sum of two parts, which can be denoted as $u(r,s)$ = $u(r',\ p,\ S)$ + $u_p(r'')$. $u(r',\ p,\ S)$ shows the utility of an instance of $r'$ at extension position $p$ in $s$, and $u_p(r'')$ is the utility of an instance of $r''$ in $s$ where the first item of such instance is after $p$. Obviously, $u_p(r'') < \textit{ru}(r',\ p,\ s)$. Then, the following can be obtained:
	\begin{align*}
	u(r,s) &= u(r',\ p,\ s) + u_p(r'')  \\
	&< u(r',\ p,\ s) + \textit{ru}(r',\ p,\ s)  \\
	&\leq max\{u(r',\ p,\ s) + \textit{ru}(r',\ p,\ s)\}  \\
	&\leq \textit{PEU}(r',s).
	\end{align*}
	Since the $q$-sequence containing $r$ must contain $r'$ as $r' \sqsubseteq r$; we have the relationship
	$\textit{pu}(r,t)$ = $\sum_{\forall s \in \textit{TD}_t \land r \sqsubseteq s}^{}u(r,s)$ $\leq$ $\sum_{\forall s \in \textit{TD}_t \land r' \sqsubseteq s}^{}{}\textit{PEU}(r',s)$ = $\textit{TPEU}(r',t)$ within the time period $t$ of $D$. Finally, $\textit{pu}(r,t)/\textit{ptsu}(t) $ = $\textit{our}(r,t)\leq \textit{TPEU}(r',t)/\textit{ptsu}(t)$.
\end{proof}

By utilizing the download closure property of the \textit{TPEU} upper bound in a time period, we designed a depth pruning strategy called local depth pruning (LDP). Given a sequence $r$ represented by a node $N$ in \textit{LS-Tree} of time period $t$, a minimum on-shelf utility threshold $\xi$, and a temporal $q$-sequence database $D$, the descendants of $N$ can be safely pruned without affecting the mining results when $\textit{TPEU(r,t)}/\textit{ptsu}(t) < \xi$. 

\begin{theorem}
	\label{TRSU}
	Given a temporal $q$-sequence database $D$ and two sequences $r$ and $r'$, suppose the node representing $r$ is the descendant of the node representing $r'$ or $r$ = $r'$, we have $$\textit{our}(r,t) \leq \textit{TRSU}(r',t)/\textit{ptsu}(t).$$
\end{theorem}

\begin{proof}
	\label{proof_TRSU}
	Let us assume that $r'$ is the extension sequence of the sequence $\alpha$; then, $\alpha$ is also a prefix of $r$. Given a $q$-sequence $s$, we have $u(t,s) \leq \textit{PEU}(\alpha,s)$ according to Proof \ref{proof_TPEU}. Based on the definition of \textit{RSU}, we have $\textit{RSU}(r',s) = \textit{PEU}(\alpha,s)$ when $s$ contains both $r'$ and $\alpha$. Then, $u(r,s) \leq \textit{RSU}(r',s)$. Conversely, if $s$ does not contain $r'$, then $u(r,s)$ = $\textit{RSU}(r',s)$ = 0, as $s$ does not contain $r$. A conclusion can be drawn that $u(r,s) \leq \textit{RSU}(r',s)$. The $q$-sequence containing $r$ contains $r'$ as $r' \sqsubseteq r$; thus, it can be obtained that $\textit{pu}(r,t)$ = $\sum_{\forall s \in \textit{TD}_t \land r \sqsubseteq s}^{}u(r,s)$ $\leq$ $\sum_{\forall s \in \textit{TD}_t \land r' \sqsubseteq s}^{}{}\textit{PEU}(r',s) $ = $\textit{RSU}(r',t)$ within the time period $t$ of $D$. Finally, $\textit{pu}(r,t)/\textit{ptsu}(t)$ = $ \textit{our}(r,t) \leq \textit{TRSU}(r',t)/\textit{ptsu}(t)$.
\end{proof}

As shown in Proof \ref{proof_TRSU}, the \textit{TRSU} upper bound demonstrates anti-monotonicity within any time periods in a temporal $q$-sequence database. Based on \textit{TRSU}, we developed a width pruning strategy called local width pruning (LWP). Given a sequence $r$ represented by a node $N$ in \textit{LS-Tree} of time period $t$, minimum on-shelf utility threshold $\xi$, and temporal $q$-sequence database $D$, OSUMS can prune the node $N$ itself and a descendant of $N$ when $\textit{TRSU(t,t)}/\textit{ptsu}(t) < \xi$. 

Varying from the upper bound \textit{TPEU}, \textit{TRSU} of a sequence $r$ is an overestimation over the utility values of itself and its descendants of $N$. The reason why we name the two pruning strategies as local strategies is that the mining processes of OSUMS are independent among time periods; consequently, they only prune in one \textit{LS-Tree} of the \textit{PLS-Forest} at a time and do not affect the search space of other time periods. 

Moreover, we also designed a novel strategy termed avoid redundancy calculations (ARC) to filter candidates when traversing the \textit{CTree} to avoid calculating actual on-shelf utility. Clearly, a sequence $r$ in \textit{CTree} may have its periodical utility values within certain time periods and have a value of zero within some others according to the field \textit{Utility}. The reason why zero occurs within the time periods is that the sequence is not on the shelf or the sequence has no possibility to be a promising osHUSP, which is pruned by the aforementioned two local pruning strategies. We denote the set of the time periods belonging to the second circumstance as $\textit{ot}_\textit{NotAppearing}(r)$.

\begin{theorem}
	Given a temporal $q$-sequence database $D$, minimum on-shelf utility $\xi$, and sequence $r$, suppose the structure \textit{CTree} has been constructed and the node $N$ representing $r$ exists in the \textit{CTree}, $r$ has no possibility of being an osHUSP and $N$ can be skipped when $(\textit{CUtility}+\sum_{\forall t \in \textit{ot}_\textit{NotAppearing}(t)}^{}\textit{ptsu}(t) \times \xi)/\sum_{\forall t \in \textit{ot}(r)}^{}\textit{ptsu}(t) < \xi$. 
\end{theorem}

\begin{proof}
	Clearly, after the first mining process, the on-shelf utility of $r$ can be divided into two parts. The first part is the sum over periodical utility values of $r$ within the time periods, denoted as $\textit{ot}_\textit{Appearing}(t)$, when the periodical utility values of $r$ have been calculated. The periodical utility values in this part are stored in the array \textit{Utility}, and the integer \textit{CUtility} is the summation over all values. The second part is the sum of the periodical utility values of $r$ within the time periods, denoted as $\textit{ot}_\textit{NotAppearing}(t)$, when $r$ has been pruned and its periodical utility value has not been calculated. The reason why $r$ has been pruned is that $r$ is not a promising osHUSP, that is, $\textit{pu}(r,t) < \textit{ptsu}(t) \times \xi$. Then, we obtain the following.

	\begin{align*}
	\textit{ou}(r)
	&= \sum_{\forall t \in \textit{ot}_\textit{Appearing}(t)}^{}\textit{pu}(r,t) + \sum_{\forall t \in \textit{ot}_\textit{NotAppearing}(t)}^{}\textit{pu}(r,t) \\
	&= \textit{CUtility} + \sum_{\forall t \in \textit{ot}_\textit{NotAppearing}(t)}^{}\textit{pu}(r,t) \\
	&\leq \textit{CUtility} + \sum_{\forall t \in \textit{ot}_\textit{NotAppearing}(t)}^{}\textit{ptsu}(t) \times \xi
	\end{align*}
	Then, the following equation can be obtained.
	\begin{align*}
	\textit{our}(r) &= \textit{ou}(r)/\sum_{\forall t \in \textit{ot}(r)}^{}\textit{ptsu}\\
	&\leq (\textit{CUtility} + \sum_{\forall t \in \textit{ot}_\textit{NotAppearing}(t)}^{}\textit{ptsu}(t) \times \xi)/\sum_{\forall t \in \textit{ot}(r)}^{}\textit{ptsu}
	\end{align*}
	
\end{proof}

Given a sequence $r$ in the \textit{CTree}, there is no requirement to calculate the actual on-shelf utility of $r$, and the node $N$ can be skipped when $(\textit{CUtility}+\sum_{\forall t \in \textit{ot}_\textit{NotAppearing}(t)}^{}\textit{ptsu}(t) \times \xi)/\sum_{\forall t \in \textit{ot}(r)}^{}\textit{ptsu}(t) < \xi$, as $r$ must not be an osHUSP. Note that the set $\textit{ot}_\textit{NotAppearing}(r)$ can be easily obtained by an AND operation over the two bit sets \textit{On-shelfTime} and \textit{CalculatedTime}. The ARC strategy can significantly reduce the high computational cost as it is time consuming to calculate the periodical utility of $r$ in a time period when starting from the beginning.

\subsubsection{Overview of OSUMS}

Based on the storage data structures and the three designed strategies, the proposed OSUMS algorithm can be stated as follows. To facilitate the statement, we present the main pseudocode of the OSUMS Algorithms \ref{alg:OSUMS1} and \ref{alg:recursive1}. 

\IncMargin{1em}
\begin{algorithm}
	\caption{OSUMS Algorithm}
	\label{alg:OSUMS1}
	\KwIn{
		$D$: a temporal $q$-sequence database;\\ 
		\textit{UT}: a utility table with external utility values of items;\\
		\textit{OT}: a table containing on-shelf information of items;\\
		$\xi$: a minimum on-shelf utility threshold.}
	\KwOut{
		\textit{osHUSP}: the complete set of osHUSPs;}
	 \BlankLine
		build a new \textit{CTree}\; 
			scan the original database $D$ to:
			(i) calculate the periodical total sequence utility (\textit{ptsu}) value of each time period;
			(ii) build the periodical $q$-matrix of each $q$-sequence in $D$\;

		\For{each time period $t \in T$}{
			scan the $q$-sequences in $\textit{TD}_t$ to: 
			(i) calculate the periodical utility value and \textit{TPEU} value of each 1-squence;
			(ii) construct the projected database of each 1-sequence in $\textit{TD}_t$\;
			\For{$r \in $ 1-sequences }{
				\If{$r$ does not exist in $CTree$}{
					construct the node representing $r$\;
				}
	
				update information of the node representing $r$ and its parent\;
				\tcp{The LDP strategy}
				\If{$\textit{TPEU}(r,t)/\textit{tpsu}(t) \geq \xi$}{
					call \textit{Local\_PGrowth($r,r.\textit{DB}_t,t$)}\;
				}
			}
		}
		\For{each node $N$ representing $r$ in \textit{CTree}}{
			calculate the periodical utilities of $r$ within the time periods when $r$ is on-shelf but has not been calculated\;
			\tcp{The ARC strategy}
			\If{$(\textit{CUtility} + \sum_{\forall t \in \textit{ot}_\textit{NotAppearing}(t)}^{}\textit{ptsu}(t) \times \xi)/\sum_{\forall t \in \textit{ot}(r)}^{}\textit{ptsu}(t) < \xi$}{
			calculate the on-shelf utility of the sequence $r$ $\textit{ou}(r)$ and the on-shelf utility ratio $\textit{our}(r)$\;
				\If{$\textit{our}(r) \ge \xi$}{
					$\textit{osHUSP} \leftarrow \textit{osHUSP} \cup r$\;
				}
			}
		}
	\Return{\textit{osHUSP}}
\end{algorithm}\DecMargin{1em}

The OSUMS method takes a temporal $q$-sequence database $D$, utility table \textit{UT}, table with on-shelf information \textit{OT}, and minimum on-shelf utility threshold $\xi$ as the inputs. The key idea of OSUMS is to enumerate the promising osHUSPs independently in each time period in the first mining process (Lines 3--14), and then identify whether they have high on-shelf utility in the second mining process (Lines 15--23). Initially, OSUMS constructs a new \textit{CTree} for storing calculation results (Lines 1). Then, OSUMS scans the original database for obtaining \textit{ptsu} values and periodical $q$-matrices (Line 2). For each time period $t$ in $T$ (Lines 3--14), OSUMS scans the $\textit{TD}_t$ to calculate essential values and construct projected databases of all 1-sequences in $t$ (Line 4). Then, \textit{CTree} will be updated according to each 1-sequence (Lines 6--9). It is to be noted that if the node representing the candidate does not appear in the tree, the node would be added as a child of the root at first (Lines 6--8). Subsequently, OSUMS adopts the LDP pruning strategy to judge whether descendants of the current 1-sequence have no possibility to be a promising osHUSP (Lines 10--12). If so, it backtracks to the root in the \textit{LS-Tree}. Otherwise, it calls the Local\_PGrowth procedure for mining longer promising osHUSPs by recursively enumerating candidate with the prefix of the current 1-sequence. In the second mining process, OSUMS traverses the \textit{CTree} to identify osHUSPs. For each node in the tree structure, OSUMS uses the AND operation to obtain the time periods when the sequence is on the shelf, but its periodical utility has not been calculated (Line 16). Then, OSUMS adopts the ARC strategy to judge whether the actual on-shelf utility must be calculated (Line 18). If $N$ is not skipped, OSUMS calculates unknown periodical utility values in $\textit{ot}_\textit{NotAppearing}$ to verify whether the sequence has on-shelf high utility (Line 19). If so, it adds the sequence to the set \textit{osHUSP} (Line 20). Finally, the final complete set of osHUSPs is discovered by the designed OSUMS algorithm (Line 24).

\IncMargin{1em}
\begin{algorithm}[htbp]
	\caption{Local\_PGrowth}
	\label{alg:recursive1}
		\KwIn{
		$r$: a sequence as prefix; \\
		$r.\textit{DB}_t$: the projected database of $t$; \\
		$t$: the current time period.}
		\BlankLine
		\For{each utility list \textit{ul} in $r.\textit{DB}_t$}{
			obtain the periodical $q$-matrix \textit{pm} of \textit{ul}\;
			scan \textit{pm} to obtain the set of $I$-Extension items \textit{ilist}\;
			scan \textit{pm} to obtain the set of $S$-Extension items \textit{slist}\;
		}
		
		\For{each item $i \in \textit{ilist}$ }{
			$r'$ $\leftarrow$ $<$$r \oplus i$$>$\; 
			\tcp{The LWP strategy}
			\If{$\textit{TRSU}(r',t)/\textit{ptsu}(t)<\xi$}{
				continue\;
			}
			construct projected database of $r'$ within the time period $t$ $r'.\textit{DB}_t$\; 
			put $r'$ into \textit{seqlist}\;
		}
		
		\For {each item $i \in \textit{slist}$ }{
			$r'$ $\leftarrow$ $<$$t \otimes s$$>$\; 
			\tcp{The LWP strategy}
			\If{$\textit{TRSU}(r',t)/\textit{ptsu}(t)<\xi$}{
				continue\;
			}
			construct projected database of $r'$ within the time period $t$ (i.e. $r'.\textit{DB}_t$)\; 
			put $r'$ into \textit{seqlist}\;
		}
		
		\For {each sequence $\textit{r} \in \textit{seqlist}$ }{
			\If{$r$ does not exist in $CTree$}{
				construct the node representing $r$\;
			}
			update information of the node representing $r$ and its parent\;
			\tcp{The LDP strategy}
			\If{$\textit{TPEU}(r,t)/\textit{tpsu}(t) \geq \xi$}{
				call \textit{Local\_PGrowth($r,r.\textit{DB}_t,t$)}\;
			}
		}	
\end{algorithm}\DecMargin{1em}

The Local\_PGrowth procedure presented in Algorithm \ref{alg:recursive1} takes a sequence $r$, the current time period $t$, and its projected database in $t$ as the inputs. Sequences with the prefix of $r$ are enumerated by applying the $I$-Extension and $S$-Extension operations. The method first scans the periodical $q$-matrix to obtain the items to be extended (Lines 1--5). For each item $i$ in \textit{ilist}, $r'$ is generated from $r$ by an $I$-Extension operation with $i$ (Line 7). Note that the \textit{TRSU} value of $r'$ in $t$ is simultaneously obtained from the scan (Lines 3--4). According to the LWP pruning strategy, OSUMS then determines whether the node representing $r$ and its descendants should be pruned (Lines 8--10). If the node has not been pruned, $r$ is saved and its projected database in $t$ is constructed (Lines 11--12). Each item in \textit{slist} can be processed in a similar manner (Lines 14--21). After extension operations, OSUMS updates the \textit{CTree} according to the calculation results of the newly generated sequences (Lines 23--26). Finally, the Local\_PGrowth procedure is recursively called for mining longer promising osHUSPs with the prefix $r$ in time period $t$ with the LDP strategy (Lines 27--29).

\subsection{OSUMS$^{+}$ Approach}

The designed pruning strategies guarantee that only on-shelf low-utility sequences are pruned in \textit{PLS-Forest}. Thus, the OSUMS method is able to discover the complete set of osHUSPs by reducing the search space to improve their performance. However, it must retain all candidates in memory in the form of the \textit{CTree} before performing the second phase of operations. This huge tree structure wastes large amounts of memory, especially when too many candidates are generated in the first mining process. Consequently, we designed a one-phase algorithm named OSUMS$^{+}$ where each candidate can be immediately identified by determining if it has on-shelf high utility once it is enumerated as a promising osHUSP. The search space of OSUMS can be represented by an \textit{LS-Tree}, as the method handles all time periods simultaneously. This subsection presents the details about the OSUMS$^{+}$ approach.

\subsubsection{Storage Data Structure}

Similar to OSUMS, the OSUMS$^{+}$ method also uses the periodical $q$-matrix to store each $q$-sequence. It facilitated the process of identifying the items to be extended and the construction of projected databases of 1-sequences. Moreover, a novel data structure with periodical utility chain is designed based on the concept of the utility chain \cite{wang2016efficiently} to represent the projected database. A periodical utility chain of a sequence is the union of utility chains in each of its on-shelf time periods. To distinguish between the utility chains generated from different time periods, we incorporated a new field \textit{Time}, which is the identifier of time periods, into the head table. As an example, Figure \ref{pdb} shows the projected database of the sequence $r$ = $<$\{\textit{a}\}, \{\textit{c}\}$>$ in the OSUMS $^{+}$ approach. We know $r$ is on the shelf within the time periods $t_2$ and $t_3$, and the utility chain within the time period $t_2$ has been illustrated in Figure \ref{db}. Thus, \textit{Time} in the head table is equal to $t_2$ in the first two utility lists. Following a recursive process, the projected databases of $l$+1-sequences can be built according to those of $l$-sequences.

\begin{figure}[htbp]
	\centering
	\includegraphics[width=0.6\linewidth]{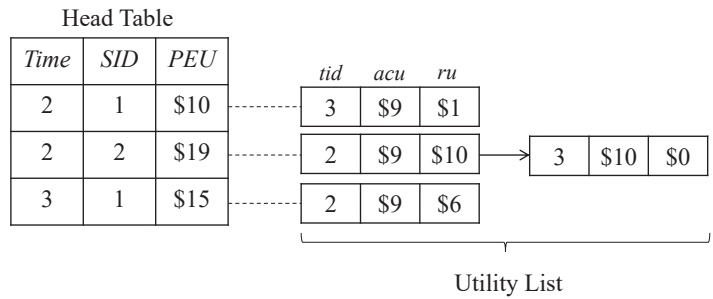}
	\caption{Projected database of $r$ = $<$\{\textit{a}\}, \{\textit{c}\}$>$ in OSUMS$^{+}$}
	\label{pdb}
\end{figure}

\subsubsection{Pruning Strategies}
OSUMS$^{+}$ adopts the \textit{PEU} and \textit{RSU} upper bounds, as well as two global pruning strategies. We defined the calculation methods of the \textit{TPEU} and \textit{TRSU} upper bounds in a time period in Section \ref{osums}. Then, we set the upper bounds for the entire on-shelf time as presented below. Given a sequence $r$ and temporal $q$-sequence database $D$, we have $\textit{TPEU}(r)$ = $\sum_{\forall t \in \textit{ot(r)}}^{}\textit{TPEU}(r,t)$ and $\textit{TRSU}(r)$ = $\sum_{\forall t \in \textit{ot(r)}}^{}\textit{TRSU}(r,t)$.

\begin{theorem}
	Given a temporal $q$-sequence database $D$ and two sequences $r$ and $r'$, suppose the node representing $r$ is a descendant of the node representing $r'$ (i.e., $r'$ is a prefix sequence of $r$), within a time period $t$ in $D$, we obtain $$\textit{our}(r) \leq \textit{TPEU}(r')/\sum_{\forall t \in \textit{ot}(r')}^{}\textit{ptsu}(t).$$
\end{theorem}

\begin{proof}
	Given a time period $t$ in $D$, we have $\textit{pu}(r,t) \leq \textit{TPEU}(r',t)$ according to Theorem \ref{TPEU}. Then, the following can be obtained 
	\begin{align*}
	\textit{ou}(r) &= \sum_{\forall t \in \textit{ot(r)}}^{}\textit{pu}(r,t) \\
	&\leq \sum_{\forall t \in \textit{ot(r')}}^{}\textit{TPEU}(r',t) \\
	& = \textit{TPEU}(r') 
	\end{align*}
	Subsequently, the following is obtained.
	\begin{align*}
	\textit{our}(r) &= \textit{ou}(r)/\sum_{\forall t \in \textit{ot}(r)}^{}\textit{ptsu}(t)\\
	&\leq \textit{TPEU}(r')/\sum_{\forall t \in \textit{ot}(r)}^{}\textit{ptsu}(t)
	\end{align*}
\end{proof}

Clearly, the \textit{TPEU} upper bound holds the download closure property of the entire time period. Then, we developed a depth pruning strategy termed global depth pruning (GDP), which is as described below. Given a sequence $r$ represented by a node $N$ in the \textit{LS-Tree}, minimum on-shelf utility threshold $\xi$, and temporal $q$-sequence database $D$, the descendants of $N$ can be safely pruned without affecting the mining results when $\textit{TPEU(r)}/ \sum_{\forall t \in \textit{ot}(r)}^{}\textit{ptsu}(t) < \xi$.

\begin{theorem}
	Given a temporal $q$-sequence database $D$ and two sequences $r$ and $r'$, suppose the node representing $r$ is the descendant of the node representing $r'$ or $r$ = $r'$, we obtain $$\textit{ou}(r) \leq \textit{TRSU}(r',t).$$
\end{theorem}

\begin{proof}
	Given a time period $t$ in $D$, we have $\textit{pu}(r,t) \leq \textit{TRSU}(r',t)$ according to Theorem \ref{TRSU}. Then, the following is obtained. 
	\begin{align*}
	\textit{ou}(r) &= \sum_{\forall t \in \textit{ot(r)}}^{}\textit{pu}(r,t) \\
	&\leq \sum_{\forall t \in \textit{ot(r')}}^{}\textit{TRSU}(r',t) \\
	& = \textit{TRSU}(r') 
	\end{align*} 
	Then,
	\begin{align*}
	\textit{our}(r) &= \textit{ou}(r)/\sum_{\forall t \in \textit{ot}(r)}^{}\textit{ptsu}(t)\\
	&\leq \textit{TRSU}(r')/\sum_{\forall t \in \textit{ot}(r)}^{}\textit{ptsu}(t). 
	\end{align*} 
\end{proof}

As proven above, the \textit{TRSU} upper bound demonstrates anti-monotonicity in the entire time periods in a temporal $q$-sequence database. Then, we designed a width pruning strategy termed global width pruning (GWP), which is described below. Given a sequence $r$ that is represented by a node $N$ in \textit{LS-Tree}, minimum on-shelf utility threshold $\xi$, and temporal $q$-sequence database $D$, $N$ and its descendants can be safely pruned when $\textit{TRSU(t)}/ \sum_{\forall t \in \textit{ot}(r)}^{}\textit{ptsu}(t) < \xi$. 

When compared to the two local pruning strategies in OSUMS, the two global strategies are able to affect all time periods. Several branches in the tree structure \textit{LS-Tree} can be efficiently and safely pruned. 

\subsubsection{Overview of OSUMS$^{+}$}

To facilitate the statement of the proposed OSUMS$^{+}$ method, we present two pieces of the pseudocode in Algorithms \ref{alg:OSUMS2} and \ref{alg:recursive2}.

\IncMargin{1em}
\begin{algorithm}[htbp]
	\caption{OSUMS$^{+}$ Algorithm}
	\label{alg:OSUMS2}
	\KwIn{
		$D$: a temporal $q$-sequence database;\\ 
		\textit{UT}: a utility table with external utility values of items;\\
		\textit{OT}: a table containing on-shelf information of items;\\
		$\xi$: a minimum on-shelf utility threshold.}
	\KwOut{
		\textit{osHUSP}: the complete set of osHUSPs;}
		\BlankLine
		first scan the original database $D$ to: 
		(i) calculate the periodical total sequence utility (\textit{ptsu}) value of each time period;
		(ii) build the periodical $q$-matrix of each $q$-sequence in $D$\;
		second scan the original database $D$ to: 
		(i) calculate the on-shelf utility and \textit{PEU} value of each 1-sequence; 
		(ii) construct the projected database of each 1-sequence\;
		\For {$r \in $ 1-sequences}{
			calculate the on-shelf utility ratio of $r$\;
			\If{$\textit{our}(r) \ge \xi$}{
				$\textit{osHUSP} \leftarrow \textit{osHUSP} \cup r$\;
			}
			\tcp{The GDP strategy}
			\If{$\textit{TPEU}(r)/\sum_{\forall t \in\textit{ot}(r)}\textit{ptsu}(t) \geq \xi$}{
				call \textit{Global\_PGrowth($r,r.\textit{DB}_t$)}\;
			}
		}
		
		\Return{\textit{osHUSP}}	
\end{algorithm}\DecMargin{1em}

As it can be observed in Algorithm \ref{alg:OSUMS2}, the input and output of OSUMS$^{+}$ are identical to those of OSUMS. In addition, both methods scan the original database at first for calculating \textit{ptsu} values and constructing the periodical $q$-matrices (Line 1). Then, OSUMS$^{+}$ performs a database scan again to obtain related information of all 1-sequences (Line 2). For each 1-sequence $r$, OSUMS$^{+}$ first calculates its on-shelf utility ratio to check whether $r$ is an osHUSP (Lines 4--7). Adopting the GDP strategy, the method can prune the descendants of $r$'s if the \textit{TPEU} of $r$ does not satisfy the given condition (Line 8). Otherwise, the function Global\_PGrowth is called to check longer sequences with the prefix of $r$ (Line 9). After the recursive process, OSUMS$^{+}$ finally returns the complete set of osHUSPs (Line 12).

\IncMargin{1em}
\begin{algorithm}[htbp]
	\caption{Global\_PGrowth}
	\label{alg:recursive2}
	\KwIn{
		$r$: a sequence as prefix; \\
		$r.\textit{DB}$: the projected database of $t$; \\}
		\BlankLine
		\For {each utility list \textit{ul} in $r.\textit{DB}$}{
			obtain the periodical $q$-matrix \textit{pm} of \textit{ul}\;
			scan \textit{pm} to obtain the set of $I$-Extension items \textit{ilist}\;
			scan \textit{pm} to obtain the set of $S$-Extension items \textit{slist}\;
		}
		
		\For {each item $i \in \textit{ilist}$ }{
			$r'$ $\leftarrow$ $<$$r \oplus i$$>$\; 
			\tcp{The GWP strategy}
			\If{$\textit{TRSU}(r')/\sum_{t \in\textit{ot}(r')}\textit{ptsu}(t) <\xi$}{
				continue\;
			}
			construct projected database of $r'$ within the time period $t$ $r'.\textit{DB}$\; 
			put $r'$ into \textit{seqlist}\;
		}
		
		\For {each item $i \in \textit{slist}$ }{
			$r'$ $\leftarrow$ $<$$t \otimes s$$>$\; 
			\tcp{The GWP strategy}
			\If{$\textit{TRSU}(r')/\sum_{\forall t \in\textit{ot}(r')}\textit{ptsu}(t) <\xi$}{
				continue\;
			}
			construct projected database of $r'$ within the time period $t$ $r'.\textit{DB}$\; 
			put $r'$ into \textit{seqlist}\;
		}
		\For {each sequence $\textit{r} \in \textit{seqlist}$ }{
			\If{$\textit{our}(r) \ge \xi$}{
				$\textit{osHUSP} \leftarrow \textit{osHUSP} \cup r$\;
			}
			\tcp{The GDP strategy}
			\If{$\textit{TPEU}(r)/\sum_{\forall t \in\textit{ot}(r)}\textit{ptsu}(t) \geq \xi$}{
				call \textit{Global\_PGrowth($r,r.\textit{DB}$)}\;
			}
		}	
\end{algorithm}\DecMargin{1em}

The Global\_PGrowth function only takes $r$ and its projected database $r.\textit{DB}$ as inputs without the time periods, as OSUMS$^{+}$ handles all time periods simultaneously. Initially, OSUMS$^{+}$ scans each utility list in $r.\textit{DB}$ to obtain items that can be extended to $r$. The processing operation of each item in OSUMS$^{+}$ bears a resemblance to that in OSUMS. The difference is that, OSUMS$^{+}$ utilizes the GDP pruning strategy to determine whether the generated sequence and its descendants must be pruned according to the $\textit{TRSU}(r')/\sum_{\forall t \in\textit{ot}(r')}\textit{ptsu}(t)$ value (Line 9/18). For each generated sequence $r'$ that has not been pruned, OSUMS$^{+}$ first checks whether it is an osHUSP (Lines 25--27). Then, if the \textit{TPEU} value of $r$ satisfies the condition, the Global\_PGrowth function is called recursively for mining patterns with the prefix of $r$ (Lines 29--31). Note that the \textit{TPEU} value of $r$ is calculated simultaneously during the process of building its projected database (Line 12/21).

%% file: 5_experiment.tex
\section{Experiments} 

This section presents sufficient experimental results to demonstrate the performance of the two proposed algorithms. The only existing algorithm TP-HOUS \cite{lan2014discovery}, which aims to mine the top-$k$ on-shelf high-utility sequences, is also implemented for the purpose of comparison. All experiments were executed on a personal computer with a 3.8 GHz Intel Core i7-10700K CPU, 32 GB RAM, and 64-bit Windows 10 operating system. The three compared algorithms were all implemented in Java JDK 1.8. Note that the mining process was forced to be terminated once its execution time exceeded 10,000 s in all experiments. Moreover, we also specified the upper limit for memory usage as 12 GB (i.e., 12,288 MB) in the IntelliJ IDEA compiler.

\subsection{Data Description}

To evaluate the performance of the algorithm, six real-world datasets were utilized in our experiment. The detailed characteristics of these datasets are listed in Table \ref{features}. Note that $|D|$ is the number of $q$-sequences, $|I|$ is the number of different $q$-items, $\textit{avg}(S)$/$\textit{max}(S)$ is the average/maximum length of $q$-sequences, \textit{\#avgItemSets} is the average number of $q$-itemsets per $q$-sequence, \textit{\#Ele} is the average number of $q$-items per $q$-itemset, \textit{avg}(\textit{UI})/\textit{max}(\textit{UI}) is the average/maximum utility of $q$-items. The first three datasets \textit{MSNBC}, \textit{Kosarak10k}, and \textit{Yoochoose} were converted from a series of click streams, which can be obtained from the MSNBC website, a Hungarian online news portal, and an ecommerce website, respectively. In addition, the linguistic datasets converted from the Bible, a famous novel Leviathan, and the sign language expression of video clips are \textit{Bible}, \textit{Leviathan}, and \textit{Sign}, respectively. In the linguistic dataset, each word can be converted into a numeric item. These six datasets represent most of the data types with various characteristics that are commonly encountered in real-world scenarios, and the use of both sequence- and item-based datasets results in a more convincing experiment. Excluding \textit{Yoochoose}\footnote{https://recsys.acm.org/recsys15/challenge/}, the other five datasets can be obtained from an open-source website\footnote{http://www.philippe-fournier-viger.com/spmf/}.

\begin{table*}[!htbp]
	\caption{Features of the datasets}
	\label{features}
	\centering      
	\begin{tabular}{|c|c|c|c|c|c|c|c|c|}
		\hline
		\textbf{Dataset} & \textbf{$|\textit{D}|$} & \textbf{$|\textit{I}|$} & \textbf{$\textit{avg}(\textit{S})$} & \textbf{$\textit{max}(\textit{S})$} & \textbf{\textit{avg}(\textit{IS})} & \textbf{\textit{\#Ele}} & \textbf{\textit{avg}(\textit{UI})} & \textbf{\textit{max}(\textit{UI})} \\ \hline 
		\textit{MSNBC} & 31,790 & 17 & 13.33 & 100 & 13.33 & 1.00 & 6.07 & 50\\ \hline     
		\textit{Kosarak10k} & 10,000 & 10,094 & 8.14 & 608 & 8.14 & 1.00 & 17.15 & 120\\ \hline
		\textit{Yoochoose} & 234,300 & 16,004 & 2.25 & 112 & 1.14 & 1.98 & 3,655 & 523,390\\ \hline
		\textit{Bible} & 36,369 & 13,905 & 21.64 & 100 & 21.64 & 1.00 & 16.29 & 140\\ \hline
		\textit{Leviathan} & 5,834 & 9,025 & 33.81 & 100 & 33.81 & 1.00 & 6.08 & 50\\ \hline
		\textit{Sign} & 730 & 267 & 52.00 & 94 & 52.00 & 1.00 & 16.71 & 110\\ \hline
	\end{tabular}
\end{table*}

\subsection{Efficiency Analysis}

In this section, we present the comparison of the performance of the two proposed algorithms in terms of runtime and memory consumption, which are common measures to evaluate the efficiency of algorithms. We conducted a series of experiments on the six datasets under various minimum on-shelf utility threshold settings. Note that the placeholder indicates that the algorithm runs so inefficiently that it satisfies the aforementioned limitation under the corresponding conditions. The details of the results are shown in Figures \ref{runtime} and \ref{memory}. Clearly, OSUMS cannot extract the desired osHUSPs in certain cases, while the OSUMS$^{+}$ worked well at all times. In general, the runtime and memory usage of methods decreases with the increase in the minimum on-shelf utility threshold, as more and more nodes in \textit{PLS-Forest} cannot be pruned by the pruning strategies. Thus, methods require more runtime to check more candidates and more memory for storing projected databases. As it can be observed, Figures \ref{runtime} and \ref{memory} convincingly verify this theory. However, the memory usage in the dataset \textit{Kosarak10k} is relatively unstable and increases and decreases sometimes. We can determine that the one-phase algorithm OSUMS$^{+}$ executes faster than OSUMS except in the case of the \textit{Yoochoose} dataset. OSUMS is more skilled in handling the database with short $q$-sequences. It is interesting to observe that the measure runtime follows an exponential decline in the six datasets. For example, the execution time of OSUMS sharply decreases when $\xi$ is less than 0.21\% in \textit{Yoochoose}. It is worth mentioning that the baseline TP-HOUS is incapable of mining the complete set of osHUSPs within a reasonable running time. This implies that our proposed pruning strategies and storage data structures improve the performance of methods crucially. In summary, OSUMS may consume a large amount of memory and is not suitable for the cases with limited memory while OSUMS$^{+}$ demonstrates wider real-life applications owing to its high efficiency.

\begin{figure*}[htbp]
	\centering
	\includegraphics[width=1\linewidth]{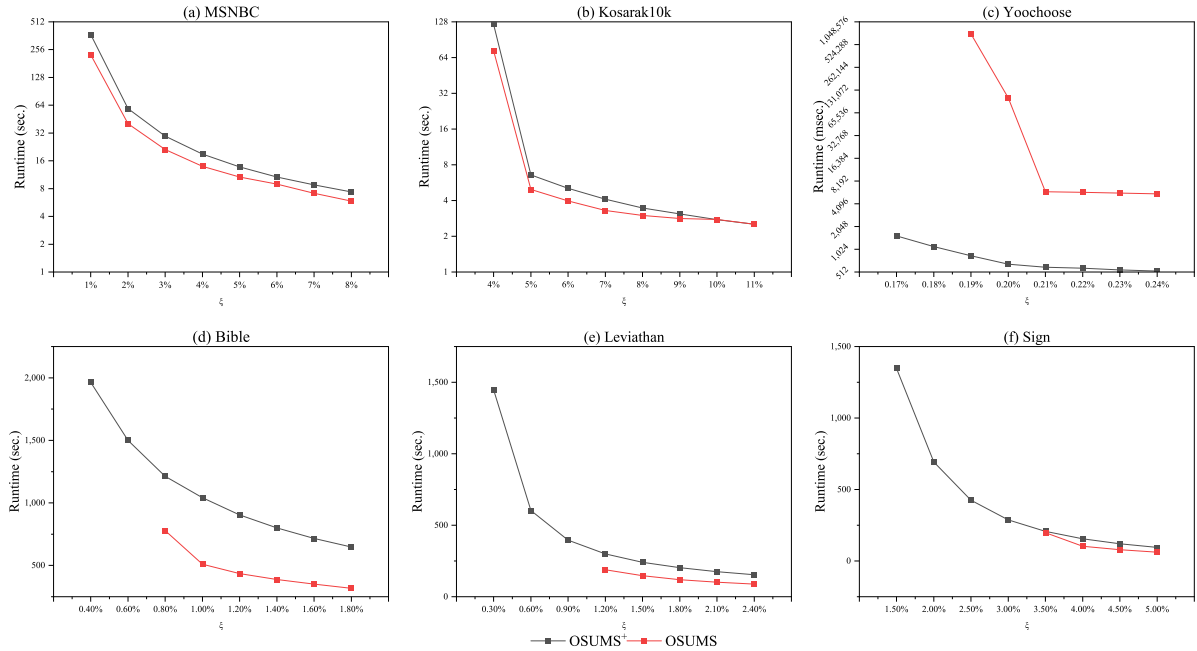}
	\caption{Runtime of the compared methods under various minimum on-shelf utility thresholds}
	\label{runtime}
\end{figure*}

\begin{figure*}[htbp]
	\centering
	\includegraphics[width=1\linewidth]{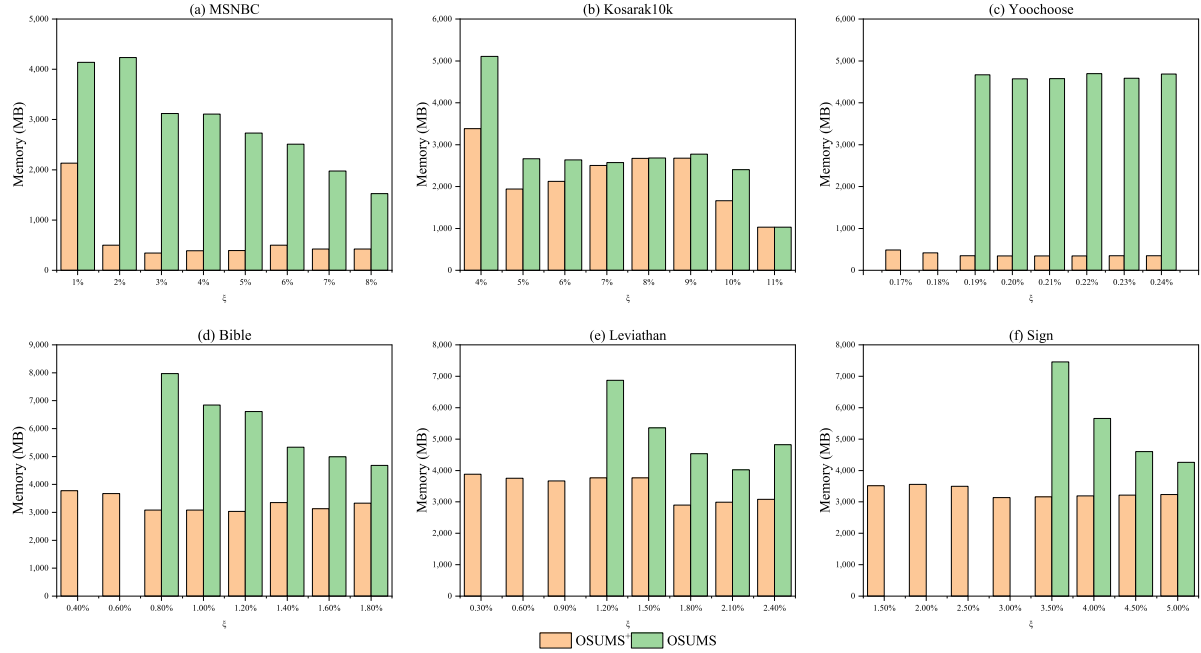}
	\caption{Memory usage of the compared approaches under varying minimum on-shelf utility thresholds}
	\label{memory}
\end{figure*}

\subsection{Candidate and Pattern Analysis}

This subsection analyzes candidates and patterns of the two proposed algorithms under varying minimum utility thresholds in the six datasets. Note that \textit{\#candidate} is the number of candidates generated, which must be checked, and \textit{\#patterns} is the number of the osHUSPs extracted by the method. As shown in Figure \ref{candidate}, with the increase in the minimum on-shelf utility threshold, the numbers of generated candidates for both algorithms follow a declining trend. In general, the descent is smooth; however, \textit{\#candidate} sharply decreases as $\xi$ crosses 4\% in the dataset \textit{Kosarak10k}. The same phenomena occurs in the mining process of OSUMS in \textit{Yoochoose}. The numbers of candidates remain relatively stable after the sharp drop point. Clearly the pruning strategies that were utilized pruned the nodes in the structure representing the search space with respect to the threshold. As we know, a deeper and wider search can be performed, which results in identifying more patterns with longer lengths. If the threshold increases, there are more candidates satisfying the inequality in the pruning strategy and they can be safely pruned. According to the definition of osHUSP, it is not difficult to understand that the number of patterns is also anti-monotonic according to the minimum on-shelf utility threshold. Moreover, the numbers of candidates and patterns are quite different among the six datasets. These depend on the utility distribution and intrinsic features of datasets. In conclusion, Figure \ref{candidate} demonstrates the pruning strategies adopted by the two algorithms to filter the candidates satisfactorily and to help in discovering osHUSPs efficiently.

\begin{figure*}[htbp]
	\centering
	\includegraphics[width=1\linewidth]{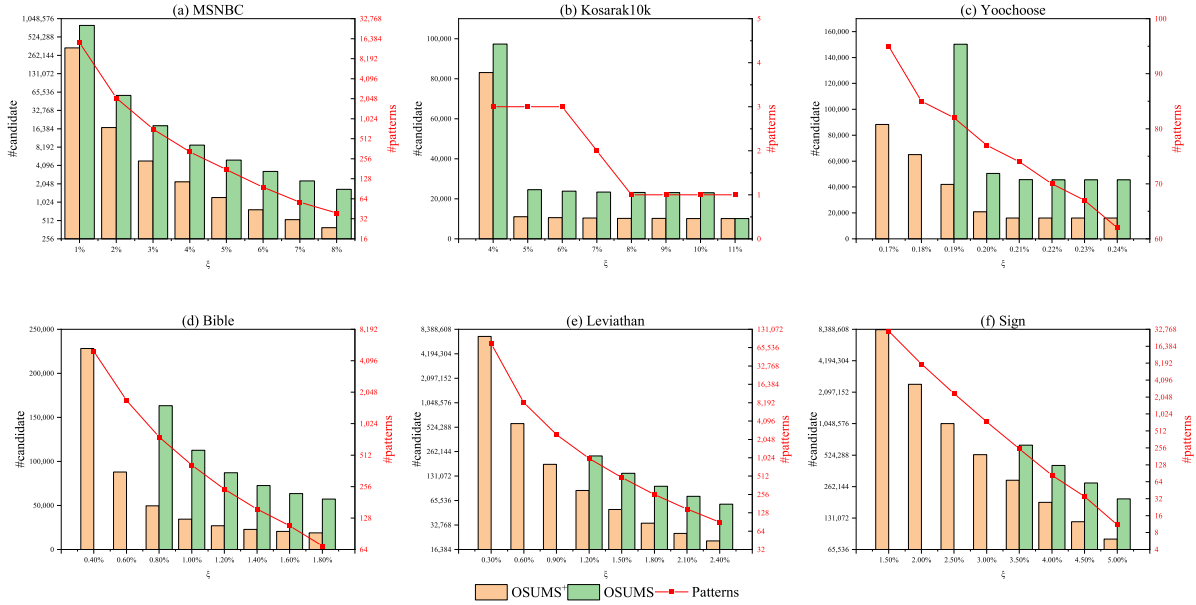}
	\caption{Experimental results of candidates and patterns generated}
	\label{candidate}
\end{figure*}

\subsection{Case Studies}

As mentioned previously, the only existing algorithm TP-HOUS is so inefficient that it cannot discover the desired osHUSPs in the limited runtime. To directly compare the two proposed methods with TP-HOUS, we performed a series of experiments on two special datasets, which were generated by extracting the first 300 $q$-sequences of \textit{MSNBC} and \textit{Yoochoose} under various minimum on-shelf utility threshold settings. To reduce the high computational complexity, each time period of the $q$-sequence ranged from 1 to 3. The results are shown in Figure \ref{case1}, which intuitively visualizes the gap between the three compared methods. Obviously, the performance of the two proposed algorithms, OSUMS and OSUMS$^{+}$, is significantly better than that of TP-HOUS in terms of runtime, memory usage, and candidate filtering. For both datasets, the execution time used by the three algorithms decreases with the increase in the minimum on-shelf utility threshold. This can be explained by the fact that the methods are required to check less candidates with the increase in threshold owing to the intrinsic poverty of pruning strategies. The results of the candidate filtering process also verify this aspect. The memory consumption of all three approaches has significant variations; however, TP-HOUS has significantly more memory requirement than our designed methods. Moreover, the phenomenon that the number of osHUSPs are anti-monotonic with respect to the threshold shows a good agreement with theoretical analysis. In conclusion, the two developed algorithms OSUMS and OSUMS$^{+}$ substantially outperform the state-of-the-art algorithm TP-HOUS.

\begin{figure*}[htbp]
	\centering
	\includegraphics[width=1\linewidth]{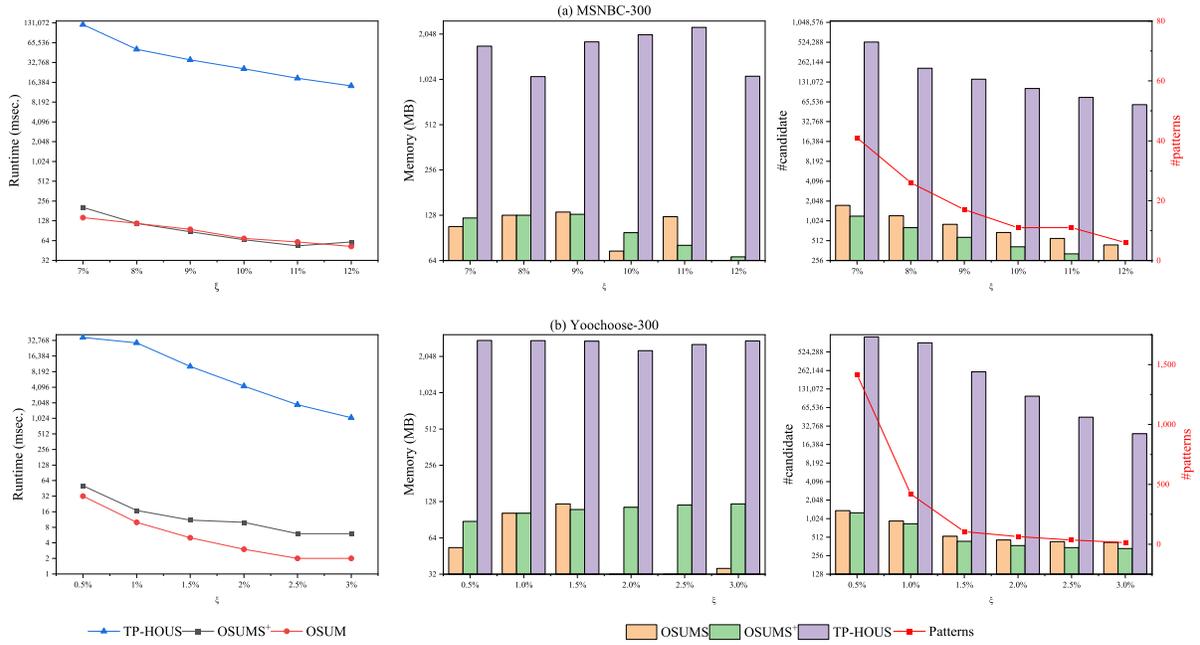}
	\caption{Results of the two cases}
	\label{case1}
\end{figure*}

\subsection{Effectiveness Analysis}

The crucial contributions of the designed methods are the pruning strategies. In theory, the pruning strategies adopted by the two proposed methods are able to vitally reduce the search space and improve the performance. To evaluate their effectiveness, we developed several method variants based on OSUMS or OSUMS$^{+}$, but without one of the pruning strategies. The performance of the compared methods in terms of runtime and memory usage for the six datasets is illustrated in Figure \ref{effectiveness}. The details are listed in Table \ref{table_effect}. For example, the OSUMS$_\textit{LDP}$ method uses OSUMS as the backbone but does not utilize the LDP pruning strategy, and so on. The minimum on-shelf utility thresholds were set to 8, 11, 0.24, 1.8, 2.4, and 5\% for the six datasets, respectively. As it can be observed from the runtime results, the two proposed methods extract the desired osHUSPs more quickly that their variants in all cases. In particular, the variants OSUMS$_\textit{LDP}$ and OSUMS$^{+}_{GDP}$ required the maximum execution time under most conditions. This indicates that the tight upper bound, called \textit{TPEU}, has a significant overestimated effect on the candidates. When compared to the two width pruning strategies, the two depth strategies demonstrated superior performance in terms of significantly reducing the search space and speeding up the mining process. In addition, it can be observed that our algorithms consume the least memory to calculate on-shelf utilities and store projected databases among the compared approaches. It is noted that the \textit{TRSU} upper bound may not be sufficiently tight in the \textit{MSNBC} dataset, as the improvement is negligible for the execution time. In addition, OSUMS$_\textit{LDP}$ ran out of memory in the \textit{SIGN} dataset. This demonstrates that the OSUMS method generated a large number of candidates without the LDP pruning strategy. Finally, a conclusion can be drawn that the designed pruning strategies contribute significantly in improving the efficiency of the algorithms.

\begin{table*}[!htbp]
	\caption{Effectiveness of pruning strategies adopting by OSUMS and OSUMS$^{+}_{}$ ("/" indicates running out of memory)}
	\label{table_effect}
	\centering
	\begin{tabular}{|c|c|p{1.5cm}<{\centering}|p{1.5cm}<{\centering}|p{1.5cm}<{\centering}|p{1.5cm}<{\centering}|p{1.5cm}<{\centering}|p{1.5cm}<{\centering}|}
		\hline \textbf{Dataset} & \textbf{Result} & $\mathbf{OSUMS}$ & $\mathbf{OSUMS_{LDP}}$ & $\mathbf{OSUMS_{LWP}}$ & $\mathbf{OSUMS^{+}}$ & $\mathbf{OSUMS^{+}_{GDP}}$ & $\mathbf{OSUMS^{+}_{GWP}}$\\
		\hline
		\multirow{2}{*}{\textit{MSNBC}} 
		& {Time} &{5.865}& {576.083} & {6.162} & {7.381}& {1,265.4} & {7.892} \\
		\cline{2-8}
		& {Memory}&{1,526} & {5,081} & {2,621} & {423}& {3,252} & {1,235} \\
		\hline
		
		\multirow{2}{*}{\textit{Kosarak10k}} 
		& {Time} &{2.529}& {21.774}& {4.949}& {2.529}& {36.977}& {7.684} \\
		\cline{2-8}
		& {Memory}&{1,027} & {2,725}& {5,862}& {1,028}& {3,109}& {1,796} \\
		\hline
		
		\multirow{2}{*}{\textit{Yoochoose}} 
		& {Time} & {5.549} & {5.163}& {5.915}& {0.528}& {1.279}& {1.995} \\
		\cline{2-8}
		& {Memory}& {4,688} & {4,708}& {4,684}& {346}& {455}& {639} \\
		\hline
		
		\multirow{2}{*}{\textit{Bible}} 
		& {Time} & {317.663} & {1,062.994}& {488.111}& {647.871}& {2,095.934}& {1,214.352} \\
		\cline{2-8}
		& {Memory}& {4,680} & {6,440}& {5,788}&{3,332}& {3,942}& {3,446} \\
		\hline
		
		\multirow{2}{*}{\textit{Leviathan}} 
		& {Time} & {87.773} & {376.834}& {134.817}& {153.624}& {732.121}& {310.936} \\
		\cline{2-8}
		& {Memory}& {4,820} & {5,970}& {6,086}& {3,081}& {3,208}& {3,241} \\
		\hline
		
		\multirow{2}{*}{\textit{Sign}} 
		& {Time} & {61.722} & {/}& {70.713}& {94.847}& {1,803.195}& {169.84} \\
		\cline{2-8}
		& {Memory}& {4,258} & {/}& {4,820}& {3,236}& {3,796}& {3,536} \\
		\hline
		
	\end{tabular}
\end{table*}

\begin{figure}[htbp]
	\centering
	\includegraphics[width=0.7\linewidth]{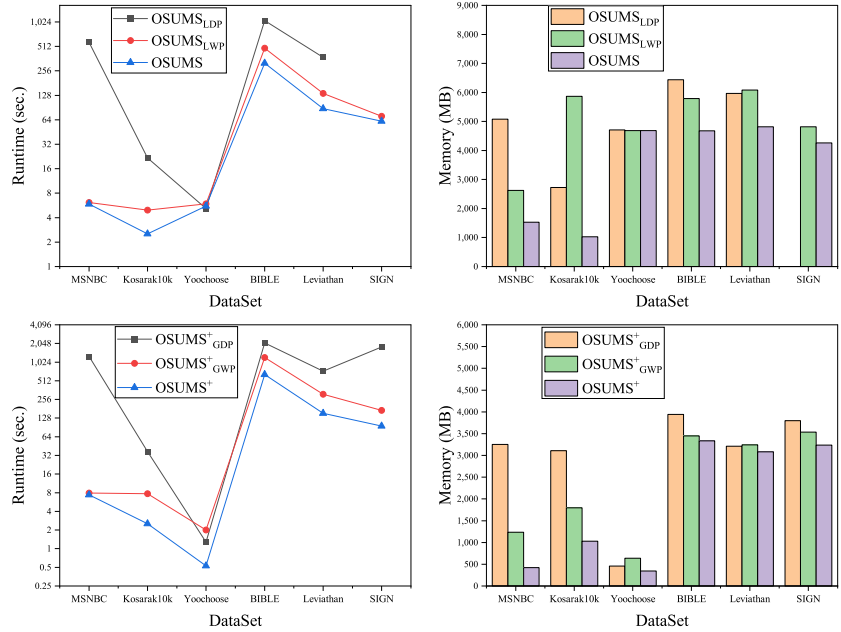}
	\caption{Effectiveness of the designed pruning strategies}
	\label{effectiveness}
\end{figure}

\subsection{Scalability}

This subsection presents the results of the scalability test for determining the robustness of the proposed two algorithms. Eight experimental datasets were obtained from \textit{Leviathan} scaled from $ 1 $ time to $ 35 $ times, that is, $ scal\_1, scal\_5, \dots, scal\_35 $, as shown on the x-coordinate of Figure \ref{scalability}. The time period of each $q$-sequence was randomly assigned an integer ranging from 1 to 5. With the same on-shelf utility minimum threshold, we conducted experiments with the OSUMS and OSUMS$^{+}$ on each dataset. The metrics of running time, memory footprint, and the number of candidates and patterns were used as measures. As it can be observed in Figure \ref{scalability}, both the runtime and memory consumption increase smoothly with increase in the dataset size. The reason is that the methods require more computational resources for utility calculation, and the projected databases grow significantly when the database is in a large scale. When compared to OSUMS, the one-phase algorithm OSUMS$^{+}$ discovered osHUSPs within a considerable amount of time, which was almost linearly related to the dataset size. Moreover, the memory usage of two algorithms also followed linear distribution. It is to be noted that the numbers of generated candidates and extracted patterns depend on the utility distribution of databases but not the size of databases. In conclusion, both the approaches have an advantage over scalability, yet OSUMS$^{+}$ is clearly the superior method.

\begin{figure*}[htbp]
	\centering
	\includegraphics[width=1\linewidth]{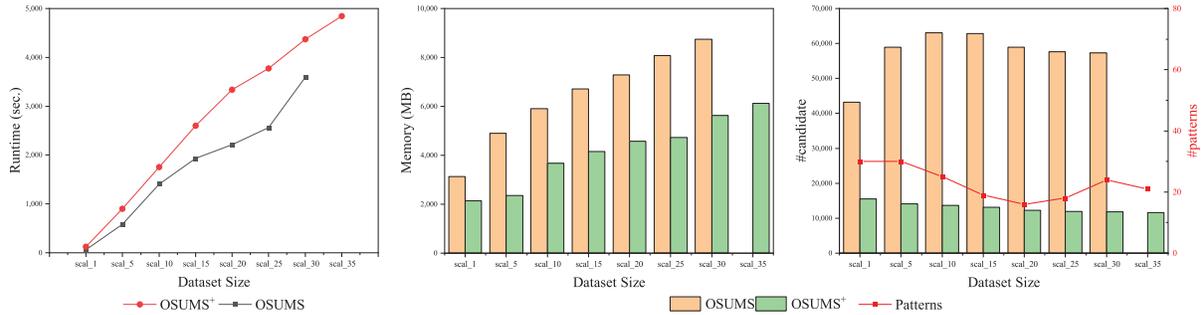}
	\caption{Scalability of the compared methods}
	\label{scalability}
\end{figure*}

%% file: 6_conclusion.tex
\section{Conclusion} 
Utility mining has performed a crucial role in the domain of knowledge discovery in database. Nevertheless, utility mining intrinsically presents a bias toward items that have longer on-shelf time as they have a greater chance to generate a high utility. To avoid the bias, the problem of OSUM was presented. In this paper, we focused on the task of OSUM in sequence data, where the sequential database is divided into several partitions according to time periods and items are associated with utilities and several on-shelf time periods. To cope with this issue, we proposed two methods, OSUMS and OSUMS$^{+}$, to extract osHUSPs efficiently. To improve the mining efficiency, we proposed two novel and compact data structures and two tight upper bounds. The first algorithm, OSUMS, operated in two phases with two local pruning strategies and an efficient calculation strategy in the mining process. For further efficiency, we designed the one-phase algorithm OSUMS$^{+}$ by adopting two global pruning strategies. Substantial experimental results on certain real and synthetic datasets show that the two methods outperformed the state-of-the-art algorithm TP-HOUS. When compared to OSUMS$^{+}$, OSUMS consumed a significantly large amount of memory and was not sufficiently efficient. In the future, an interesting direction is to redesign the two algorithms and develop a parallel and distributed version, for example, utilizing the MapReduce paradigm to support OSUM in sequence data on large-scale databases in distributed environments.